\documentclass[dvipsnames]{article}

\usepackage[utf8]{inputenc}

\usepackage{amsmath}
\usepackage{amssymb}

\usepackage{amsthm}

\usepackage{xspace}

\usepackage{colortbl}
    
\usepackage{booktabs}

\usepackage{mathtools}

\usepackage{tikz}
\usetikzlibrary{arrows,decorations.text,decorations.pathmorphing,decorations.pathreplacing,positioning,patterns,tikzmark,arrows.meta,backgrounds,shapes.geometric}
\tikzset{  
  >=stealth',
  plainnode/.style 	= {draw, thick,circle, minimum size=8mm, inner sep=0mm}
}

\newcommand{\nats}{\mathbb{N}}
\renewcommand{\epsilon}{\varepsilon}
\renewcommand{\phi}{\varphi}
\newcommand{\size}[1]{|#1|}
\newcommand{\pow}[1]{2^{#1}}
\newcommand{\cceq}{\mathop{::=}}
\newcommand{\set}[1]{\{#1\}}

\newcommand{\F}{\mathop{\mathbf{F}\vphantom{a}}\nolimits}
\newcommand{\G}{\mathop{\mathbf{G}\vphantom{a}}\nolimits}
\DeclareMathOperator{\U}{\mathbf{U}}
\newcommand{\X}{\mathop{\mathbf{X}\vphantom{a}}\nolimits}

\newcommand{\ltl}{\mathrm{LTL}\xspace}

\newcommand{\hyltl}{\mathrm{Hyper\-LTL}\xspace}
\newcommand{\hyqptl}{\mathrm{Hyper\-QPTL}\xspace}
\newcommand{\hyqptlplus}{\mathrm{Hyper\-QPTL^+}\xspace}
\newcommand{\sohyltl}{\mathrm{Hyper^2LTL}\xspace}
\newcommand{\sohyltlfp}{\mathrm{Hyper^2LTL_{\mathrm{mm}}}\xspace}
\newcommand{\sohyltlfpmax}{\mathrm{Hyper^2LTL_{\mathrm{mm}}^{\largest}}\xspace}
\newcommand{\sohyltlfpmin}{\mathrm{Hyper^2LTL_{\mathrm{mm}}^{\smallest}}\xspace}
\newcommand{\sohyltlfpsmalar}{\mathrm{Hyper^2LTL_{\mathrm{mm}}^{\smalar}}\xspace}
\newcommand{\lfpsohyltlfp}{\mathrm{lfp\text{-}{Hyper^2LTL_{\mathrm{mm}}}}\xspace}

\newcommand{\hyctlstar}{\mathrm{HyperCTL^*}\xspace}
\newcommand{\sohyltlfpold}{\mathrm{Hyper^2LTL_{\mathrm{fp}}}\xspace}

\newcommand{\suffix}[2]{#1[#2,\infty)}

\newcommand{\fovar}{\mathcal{V}_1}
\newcommand{\sovar}{\mathcal{V}_2}
\newcommand{\ap}[0]{\mathrm{AP}}

\newcommand{\univar}{X_a}
\newcommand{\unidisvar}{X_d}

\newcommand{\smallest}{\curlyvee}
\newcommand{\largest}{\curlywedge}
\newcommand{\smalar}{{\mathrlap{\curlyvee}\curlywedge}}
\newcommand{\solutions}{\mathrm{sol}}

\newcommand{\pspace}{\textsc{PSpace}\xspace}

\newcommand{\tower}{\textsc{Tower}\xspace}

\newcommand{\myquot}[1]{``#1''}

\newcommand{\tsys}{\mathfrak{T}}
\newcommand{\traces}{\mathrm{Tr}}
\newcommand{\initmark}{I}

\newcommand{\posprop}{\texttt{+}}
\newcommand{\negprop}{\texttt{-}}
\newcommand{\setprop}{\texttt{s}}
\newcommand{\inprop}{\texttt{x}}

\newcommand{\prop}{\texttt{p}}
\newcommand{\argone}{\texttt{arg1}}
\newcommand{\argtwo}{\texttt{arg2}}
\newcommand{\res}{\texttt{res}}
\newcommand{\add}{\texttt{add}}
\newcommand{\mult}{\texttt{mult}}

\newcommand{\plustimes}{{(+,\cdot)}}
\newcommand{\hyperize}{{\mathit{hyp}}}
\newcommand{\arithmetize}{{\mathit{ar}}}
\newcommand{\allsets}{{\mathit{allSets}}}
\newcommand{\pair}{\mathit{pair}}
\newcommand{\istrace}{\mathit{isTrace}}

\newcommand{\arith}{\mathit{arith}}

\newcommand{\alltraces}{\mathit{allTraces}}
\newcommand{\parti}{\mathit{part}}
\newcommand{\update}{\mathit{update}}

\newcommand{\ispath}{\mathit{isPath}}

\newcommand{\prefs}{\mathit{Prefs}}

\newcommand{\natsstruct}{(\nats, +, \cdot, <, \in)}

\newcommand{\con}{\mathrm{con}}
\newcommand{\step}{\mathrm{step}}
\newcommand{\tracein}{\triangleright}

\newcommand{\quant}{Q}
\newcommand{\overdot}[1]{\dot{#1}}
\newcommand{\overdotdot}[1]{\ddot{#1}}
\newcommand{\cw}{\mathrm{cw}}

\newcommand{\lfp}{\mathrm{lfp}}

\newcommand{\equals}[3]{#1 =_{#3} #2}

\newcommand{\marker}{\texttt{m}}
\newcommand{\complete}{\mathit{complete}}
\newcommand{\merge}{^\smallfrown}
\newcommand{\guard}{\mathit{guard}}
\newcommand{\enc}{\mathit{enc}}
\newcommand{\extend}{\mathit{ext}}

\newcommand{\hastree}{\mathit{hasTree}}
\newcommand{\encode}[1]{\langle #1\rangle}

\newcommand{\vertices}{\mathit{V}}
\newcommand{\edges}{\mathit{E}}
\newcommand{\edge}{\mathit{edge}}
\newcommand{\prefix}{\mathit{prefix}}
\newcommand{\prefixes}{\mathit{prefixes}}
\newcommand{\rel}{\mathit{rel}}
\newcommand{\init}{\mathit{init}}
\newcommand{\consistent}{\mathit{cons}}

\newtheorem{remark}{Remark}
\newtheorem{lemma}{Lemma}
\newtheorem{theorem}{Theorem}
\newtheorem{proposition}{Proposition}

\newtheorem{corollary}{Corollary}
\usepackage{fullpage}

\newcommand{\texorpdfstring}[2]{#1}

\title{The Complexity of Fragments of Second-Order {HyperLTL}\thanks{Supported by DIREC – Digital Research Centre Denmark and the European Union. 

\includegraphics[scale=.15]{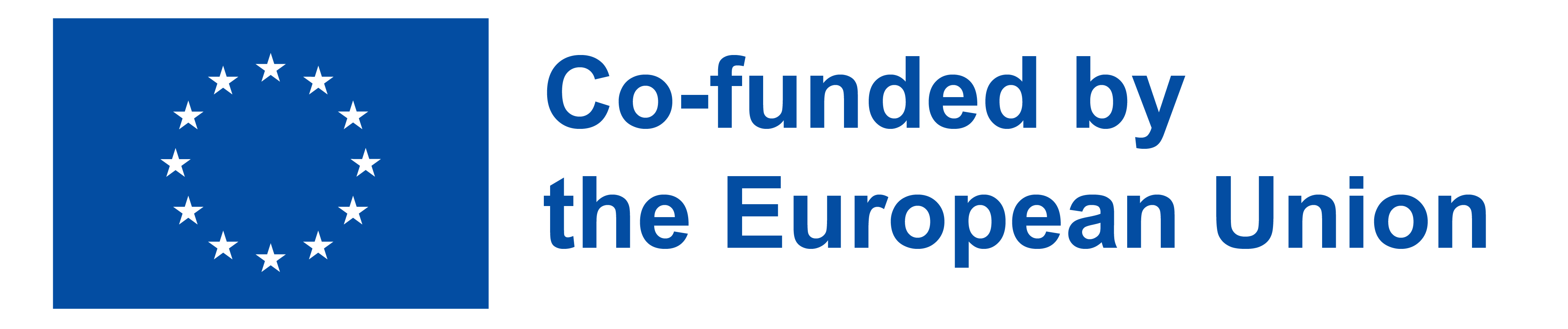}}}

\author{
Gaëtan Regaud (ENS Rennes, Rennes, France)\\
Martin Zimmermann (Aalborg University, Aalborg, Denmark)}
\date{}

\begin{document}

\maketitle

\begin{abstract}
We settle the complexity of satisfiability, finite-state satisfiability, and model-checking for several fragments of second-order HyperLTL, which extends HyperLTL with quantification over sets of traces: they are all in the analytical hierarchy and beyond.
\end{abstract}

\section{Introduction}
\label{sec:intro}

The temporal logic~$\hyltl$~\cite{ClarksonFKMRS14} has been introduced to express hyperproperties~\cite{ClarksonS10}, i.e., properties that involve multiple execution traces of a system. 
Many important information-flow specifications are hyperproperties, and expressible in $\hyltl$~\cite{FinkbeinerRS15}.
As $\hyltl$ model-checking is decidable~\cite{ClarksonFKMRS14}, this enables the automated verification of a wide range of expressive specifications. 

On a technical level, $\hyltl$ extends the classical linear-time logic $\ltl$~\cite{Pnueli77} (which is evaluated over a single execution trace) by quantification over traces, and is therefore evaluated over a set of traces. 
However, there are also important hyperproperties that are \emph{not} expressible in $\hyltl$ (nor its branching-time cousin $\hyctlstar$~\cite{ClarksonFKMRS14}), e.g., asynchronous hyperproperties and common knowledge in multi-agent systems. 
Trace quantification alone is not sufficient to express these properties.

This motivated Beutner et al.~\cite{BeutnerFFM24} to study $\hyltl$ extended by second-order quantification, i.e., quantification over sets of traces. In particular, they showed that model-checking the resulting logic, termed $\sohyltl$, is $\Sigma_1^1$-hard, i.e., highly-undecidable.
This high complexity led them to study fragments with restricted second-order quantification that are still able to express asynchronous hyperproperties and common knowledge.
In particular, they introduced
\begin{itemize}
    \item $\sohyltlfp$, which restricts second-order quantification to maximal/minimal sets satisfying a guard formula\footnote{In~\cite{DBLP:conf/cav/BeutnerFFM23}, this fragment is termed $\sohyltlfpold$.}, and 
    \item $\lfpsohyltlfp$, which restricts second-order quantification to least fixed points of $\hyltl$-definable operators. Note that in this fragment, quantification degenerates to the \myquot{computation} of unique fixed points. 
\end{itemize}
The main result of Beutner et al.\ is a partial model-checking algorithm for $\lfpsohyltlfp$: Their algorithm over- and underapproximates fixed points and then
invokes a $\hyltl$ model-checking algorithm on these approximations. A prototype implementation of the algorithm is able to model-check properties capturing asynchronous hyperproperties and common knowledge, among other properties.

Beutner et al.'s focus was on algorithmic aspects, e.g., model-checking and monitoring~\cite{BeutnerFFM24}.
On the other hand, Frenkel and Zimmermann studied the complexity of the most important verification problems for $\sohyltl$ and its fragments, i.e., satisfiability, finite-state satisfiability, and model-checking~\cite{sohypercomplexity}. 
They showed that all three are equivalent to truth in third-order arithmetic, i.e., predicate logic with quantification over natural numbers, sets of natural numbers, and sets of sets of natural numbers.
These results hold both for full $\sohyltl$ and the fragment~$\sohyltlfp$ and for two different semantics:
the standard semantics introduced by Beutner et al.\ (where set quantifiers range over arbitrary sets) and closed-world semantics (where second-order quantifiers range only over subsets of the model).
Furthermore, they showed that $\lfpsohyltlfp$ has better properties, even though all three problems are still highly undecidable:
$\lfpsohyltlfp$ satisfiability under closed-world semantics is $\Sigma_1^1$-hard, i.e., not harder than $\hyltl$ satisfiability~\cite{hyperltlsatconf}, and finite-state satisfiability and model-checking are $\Sigma_1^1$-hard and in $\Sigma_2^2$ (under both semantics). 

Thus, there are several gaps in their results, i.e., the complexity of finite-state satisfiability and model-checking for $\lfpsohyltlfp$ and the complexity of $\lfpsohyltlfp$ satisfiability under standard semantics.
Finally, their results for $\sohyltlfp$ rely on using both minimal and maximal sets satisfying guards, i.e., the complexity of the fragments using only one type of \myquot{polarity} is still open. 

\subsection{Our Contributions}
In this work, we settle the complexity of all these open problems: $\lfpsohyltlfp$ satisfiability under standard semantics is $\Sigma_1^2$-complete, finite-state satisfiability and model-checking for $\lfpsohyltlfp$ are equivalent to truth in second-order arithmetic (under both semantics) and all three problems are equivalent to truth in third-order arithmetic for both unipolar fragments of $\sohyltlfp$ (again under both semantics).
Table~\ref{table_results} lists our results and compares them to related logics.

\begin{table}[h]
    \centering
    \caption{List of our results (in bold and red) and comparison to related logics~\cite{FinkbeinerH16,hyperltlsat,sohypercomplexity,regaud2024complexityhyperqptl}. \myquot{T2A-equivalent} stands for \myquot{equivalent to truth in second-order arithmetic}, \myquot{T3A-equivalent} for \myquot{equivalent to truth in third-order arithmetic}. Unless explicitly specified, results hold for both semantics.}
    \renewcommand{\arraystretch}{1.1}
    \begin{tabular}{llll}
      \small   Logic &\small  Satisfiability  &\small  Finite-state satisfiability  &\small  Model-checking \\
         \midrule
        $\ltl$ &  \pspace-complete &  \pspace-complete &  \pspace-complete\\
        \rowcolor{lightgray!40}$\hyltl$ & $\Sigma_1^1$-complete  & $\Sigma_1^0$-complete  & \tower-complete \\
        $\sohyltl$ &  T3A-equivalent &  T3A-equivalent &  T3A-equivalent\\
        \rowcolor{lightgray!40}$\sohyltlfp$ &  T3A-equivalent &  T3A-equivalent &  T3A-equivalent\\
        $\sohyltlfpmax$ & \textcolor{Maroon}{\textbf{T3A-equivalent}} &  \textcolor{Maroon}{\textbf{T3A-equivalent}} &  \textcolor{Maroon}{\textbf{T3A-equivalent}}\\
        \rowcolor{lightgray!40}$\sohyltlfpmin$ &  \textcolor{Maroon}{\textbf{T3A-equivalent}} &  \textcolor{Maroon}{\textbf{T3A-equivalent}} &  \textcolor{Maroon}{\textbf{T3A-equivalent}}\\
        $\lfpsohyltlfp$ &  \textcolor{Maroon}{\bf  \boldmath$\Sigma_1^2$-complete}/ &  \textcolor{Maroon}{\textbf{T2A-equivalent}} &  \textcolor{Maroon}{\textbf{T2A-equivalent}}\\
                        & $\Sigma_1^1$-complete (CW) &&\\
        \rowcolor{lightgray!40}$\hyqptl$ &  $\Sigma_1^2$-complete & $\Sigma_1^0$-complete & $\tower$-complete\\
        $\hyqptlplus$ & T3A-equivalent &  T3A-equivalent &  T3A-equivalent\\
    \end{tabular}
    \label{table_results}
\end{table}

\section{Preliminaries}
\label{sec:prels}

We denote the nonnegative integers by $\nats$. 
An alphabet is a nonempty finite set. 
The set of infinite words over an alphabet~$\Sigma$ is denoted by $\Sigma^\omega$.
Let $\ap$ be a nonempty finite set of atomic propositions. 
A trace over $\ap$ is an infinite word over the alphabet~$\pow{\ap}$.
Given a subset~$\ap' \subseteq \ap$, the $\ap'$-projection of a trace~$t(0)t(1)t(2) \cdots$ over $\ap$ is the trace~$(t(0) \cap \ap')(t(1) \cap \ap')(t(2) \cap \ap') \cdots$ over $\ap'$.
The $\ap'$-projection of $T \subseteq (\pow{\ap})^\omega$ is defined as the set of $\ap$-projections of traces in $T$.
Now, let $\ap$ and $\ap'$ be two disjoint sets, let $t$ be a trace over~$\ap$, and let $t'$ be a trace over $\ap'$. 
Then, we define $t \merge t'$ as the pointwise union of $t$ and $t'$, i.e., $t \merge t'$ is the trace over $\ap \cup \ap'$ defined as $(t(0) \cup t'(0))(t(1) \cup t'(1))(t(2) \cup t'(2))\cdots$.

A transition system~$\tsys = (V,E,I, \lambda)$ consists of a finite nonempty set~$V$ of vertices, a set~$E \subseteq V \times V$ of (directed) edges, a set~$I \subseteq V$ of initial vertices, and a labeling~$\lambda\colon V \rightarrow \pow{\ap}$ of the vertices by sets of atomic propositions.
We assume that every vertex has at least one outgoing edge.
A path~$\rho$ through~$\tsys$ is an infinite sequence~$\rho(0)\rho(1)\rho(2)\cdots$ of vertices with $\rho(0) \in I$ and $(\rho(n),\rho(n+1))\in E$ for every $n \ge 0$.
The trace of $\rho$ is defined as $ \lambda(\rho ) = \lambda(\rho(0))\lambda(\rho(1))\lambda(\rho(2))\cdots$.
The set of traces of $\tsys$ is $\traces(\tsys) = \set{\lambda(\rho) \mid \rho \text{ is a path through $\tsys$}}$.

\subparagraph*{\texorpdfstring{\boldmath$\sohyltl$}{Second-order HyperLTL}.}
\label{subsec_hyperltl}

Let $\fovar$ be a set of first-order trace variables (i.e., ranging over traces) and $\sovar$ be a set of second-order trace variables (i.e., ranging over sets of traces) such that $\fovar \cap \sovar = \emptyset$.
We typically use $\pi$ (possibly with decorations) to denote first-order variables and $X, Y, Z$ (possibly with decorations) to denote second-order variables.
Also, we assume the existence of two distinguished second-order variables~$\univar, \unidisvar  \in \sovar$ such that $\univar$ refers to the set~$(\pow{\ap})^\omega$ of all traces, and $\unidisvar$ refers to the universe of discourse (the set of traces the formula is evaluated over).

The formulas of $\sohyltl$ are given by the grammar
\begin{align*}
\phi  {} \cceq {} \exists X.\ \phi \mid \forall X.\ \phi \mid \exists \pi \in X.\ \phi \mid \forall \pi \in X.\ \phi \mid \psi \quad\quad
\psi {}  \cceq {} \prop_\pi \mid \neg \psi \mid \psi \vee \psi \mid \X \psi \mid \psi \U \psi    
\end{align*}
where $\prop$ ranges over $\ap$, $\pi$ ranges over $\fovar$, $X$ ranges over $\sovar$, and $\X$ (next) and $\U$ (until) are the temporal operators. Conjunction~($\wedge$), exclusive disjunction~$(\oplus)$, implication~($\rightarrow$), and equivalence~$(\leftrightarrow)$ are defined as usual, and the temporal operators eventually~($\F$) and always~($\G$) are derived as $\F\psi = \neg \psi\U \psi$ and $\G \psi = \neg \F \neg \psi$. We measure the size of a formula by its number of distinct subformulas.

The semantics of $\sohyltl$ is defined with respect to a variable assignment, i.e., a partial mapping~$\Pi \colon \fovar \cup \sovar \rightarrow (\pow{\ap})^\omega \cup \pow{(\pow{\ap})^\omega}$
such that
\begin{itemize}
    \item if $\Pi(\pi)$ for $\pi\in\fovar$ is defined, then $\Pi(\pi) \in (\pow{\ap})^\omega$ and
    \item if $\Pi(X)$ for $X \in\sovar$ is defined, then $\Pi(X) \in \pow{(\pow{\ap})^\omega}$.
\end{itemize}
Given a variable assignment~$\Pi$, a variable~$\pi \in \fovar$, and a trace~$t$, we denote by $\Pi[\pi \mapsto t]$ the assignment that coincides with $\Pi$ on all variables but $\pi$, which is mapped to $t$. 
Similarly, for a variable~$X \in \sovar$ and a set~$T$ of traces, $\Pi[X \mapsto T]$ is the assignment that coincides with $\Pi$ everywhere but $X$, which is mapped to $T$.
Furthermore, $\suffix{\Pi}{j}$ denotes the variable assignment mapping every $\pi \in \fovar$ in $\Pi$'s domain to $\Pi(\pi)(j)\Pi(\pi)(j+1)\Pi(\pi)(j+2) \cdots $, the suffix of $\Pi(\pi)$ starting at position $j$ (the assignment of variables~$X \in \sovar$ is not updated). 

For a variable assignment~$\Pi$ we define 
\begin{itemize}
	\item $\Pi \models \prop_\pi$ if $\prop \in \Pi(\pi)(0)$,
	\item $\Pi \models \neg \psi$ if $\Pi \not\models \psi$,
	\item $\Pi \models \psi_1 \vee \psi_2 $ if $\Pi \models \psi_1$ or $\Pi  \models \psi_2$,
	\item $\Pi \models \X \psi$ if $\suffix{\Pi}{1} \models \psi$,
	\item $\Pi \models \psi_1 \U \psi_2$ if there exists a $j \ge 0$ such that $\suffix{\Pi}{j} \models \psi_2$ and for all $0 \le j' < j$ we have $ \suffix{\Pi}{j'} \models \psi_1$ ,
	\item $\Pi \models \exists \pi \in X.\ \phi$ if there exists a trace~$t \in \Pi(X)$ such that $\Pi[\pi \mapsto t] \models \phi$ ,
	\item $\Pi \models \forall \pi \in X.\ \phi$ if for all traces~$t \in \Pi(X)$ we have $\Pi[\pi \mapsto t] \models \phi$,
    \item $\Pi \models \exists X.\ \phi$ if there exists a set~$T \subseteq (\pow{\ap})^\omega$ such that $\Pi[X\mapsto T] \models \phi$, and
    \item $\Pi \models \forall X.\ \phi$ if for all sets~$T \subseteq (\pow{\ap})^\omega$ we have $\Pi[X\mapsto T] \models \phi$.
\end{itemize}

A sentence is a formula in which only the variables $\univar,\unidisvar$ can be free. 
The variable assignment with empty domain is denoted by $\Pi_\emptyset$. 
We say that a set~$T$ of traces satisfies a $\sohyltl$ sentence~$\phi$, written $T \models \phi$, if $\Pi_\emptyset[\univar \mapsto (\pow{\ap})^\omega, \unidisvar \mapsto T]\models \phi$, i.e., if we assign the set of all traces to~$\univar$ and the set~$T$ to the universe of discourse~$\unidisvar$.
In this case, we say that $T$ is a model of $\varphi$.
A transition system~$\tsys$ satisfies $\phi$, written $\tsys \models \phi$, if $\traces(\tsys)\models \phi$. 

Although $\sohyltl$ sentences are required to be in prenex normal form, $\sohyltl$ sentences are closed under Boolean combinations, which can easily be seen by transforming such a sentence into an equivalent one in prenex normal form (which might require renaming of variables).
Thus, in examples and proofs we will often use Boolean combinations of $\sohyltl$ sentences. 

\begin{remark}
\label{rem_hyltlisfragment}
$\hyltl$ is the fragment of $\sohyltl$ obtained by disallowing second-order quantification and only allowing first-order quantification of the form~$\exists \pi \in \unidisvar$ and $\forall \pi \in \unidisvar$, i.e., one can only quantify over traces from the universe of discourse.
Hence, we typically simplify our notation to $\exists\pi$ and $\forall\pi$ in $\hyltl$ formulas.
\end{remark}

Throughout the paper, we use the following shorthands to simplify our formulas:
\begin{itemize}
    
    \item We write $\equals{\pi}{\pi'}{\ap'}$ for a set~$\ap' \subseteq \ap$ for the formula~$ \G\bigwedge_{\prop\in\ap'}(\prop_{\pi} \leftrightarrow \prop_{\pi'})$ expressing that the $\ap'$-projection of $\pi$ and the $\ap'$-projection of $\pi'$ are equal.
    
    \item We write $\pi\tracein X $  for the formula~$\exists \pi' \in X.\ \equals{\pi}{\pi'}{\ap}$ expressing that the trace~$\pi$ is in $X$. Note that this shorthand cannot be used under the scope of temporal operators, as we require formulas to be in prenex normal form.

\end{itemize}

\subparagraph*{Closed-World Semantics.}
Second-order quantification in $\sohyltl$ as defined by Beutner et al.~\cite{DBLP:conf/cav/BeutnerFFM23} (and introduced above) ranges over arbitrary sets of traces (not necessarily from the universe of discourse) and first-order quantification ranges over elements in such sets, i.e., (possibly) again over arbitrary traces.
\emph{Closed-world} semantics for $\sohyltl$, introduced by Frenkel and Zimmermann~\cite{sohypercomplexity}, disallow this: Formulas may not use the variable~$\univar$ and the semantics of set quantifiers are changed as follows, 
where the closed-world semantics of atomic propositions, Boolean connectives, temporal operators, and trace quantifiers is defined as before:
\begin{itemize}
    \item $\Pi \models_\cw \exists X.\ \phi$ if there exists a set~$T \subseteq \Pi(\unidisvar)$ such that $\Pi[X\mapsto T] \models_\cw \phi$, and 
    \item $\Pi \models_\cw \forall X.\ \phi$ if for all sets~$T \subseteq \Pi(\unidisvar)$ we have $\Pi[X\mapsto T] \models_\cw \phi$.
\end{itemize}
We say that $T \subseteq (\pow{\ap})^\omega$ satisfies $\phi$ under closed-world semantics, if $\Pi_\emptyset[\unidisvar \mapsto T] \models_\cw \phi$.
Hence, under closed-world semantics, second-order quantifiers only range over subsets of the universe of discourse. 
Consequently, first-order quantifiers also range over traces from the universe of discourse. 

\subparagraph*{Arithmetic.}
\label{subsec_arithmetic}

To capture the complexity of undecidable problems, we consider formulas of arithmetic, i.e., predicate logic with signature~$(+, \cdot, <, \in)$, evaluated over the structure~$\natsstruct$. 
A type~$0$ object is a natural number in $\nats$, a type~$1$ object is a subset of $\nats$, and a type~$2$ object is a set of subsets of $\nats$.

First-order arithmetic allows to quantify over type~$0$ objects, second-order arithmetic allows to quantify over type~$0$ and type~$1$ objects, and third-order arithmetic allows to quantify over type~$0$, type~$1$, and type~$2$ objects.
Note that every fixed natural number is definable in first-order arithmetic, so we freely use them as syntactic sugar. Similarly, equality can be eliminated if necessary, as it can be expressed using~$<$.

Truth in second-order arithmetic is the following problem: given a sentence~$\phi$ of second-order arithmetic, does $\natsstruct $ satisfy $\phi$?
Truth in third-order arithmetic is defined analogously.
Furthermore, arithmetic formulas with a single free first-order variable define sets of natural numbers. We are interested in the classes
\begin{itemize}
    \item $\Sigma_1^1$ containing sets of the form~$\set{x \in\nats \mid \exists X_1 \subseteq \nats.\ \cdots \exists X_k\subseteq \nats.\ \psi(x, X_1, \ldots,X_k )}$, where $\psi$ is a formula of arithmetic with arbitrary quantification over type~$0$ objects (but no other quantifiers), and
    \item $\Sigma_1^2$ containing sets of the following form,
    where $\psi$ is a formula of arithmetic with arbitrary quantification over type~$0$ and type~$1$ objects (but no other quantifiers):
    $\set{x \in\nats \mid \exists \mathcal{X}_1\subseteq \pow{\nats}.\ \cdots \exists \mathcal{X}_k\subseteq \pow{\nats}. \ \psi(x, \mathcal{X}_1, \ldots,\mathcal{X}_k)}.$
\end{itemize}

\subparagraph*{Problem Statement.}

We are interested in the complexity of the following three problems for fragments of $\sohyltl$ for both semantics:
\begin{itemize}
    \item Satisfiability: Given a sentence~$\phi$, does it have a model, i.e., is there a set~$T$ of traces such that $T \models \phi$?
    \item Finite-state satisfiability: Given a sentence~$\phi$, is it satisfied by a transition system, i.e., is there a transition system~$\tsys$ such that $\tsys \models \phi$?
    \item Model-checking: Given a sentence~$\phi$ and a transition system~$\tsys$, do we have $\tsys \models \phi$?
\end{itemize}

It is known~\cite{sohypercomplexity} that all three problems are equivalent to truth in third-order arithmetic for full $\sohyltl$ and for the fragment where second-order quantification is restricted to maximal and minimal sets satisfying a guard formula (see Subsection~\ref{subsec_mmsands} for a formal definition).
Furthermore, if one restricts second-order quantification even further, i.e., to least fixed points of $\hyltl$-definable operators (see Subsection~\ref{subsec_lfpsands} for a formal definition), then the complexity finally decreases: satisfiability under closed-world semantics is $\Sigma_1^1$-complete while finite-state satisfiability and model-checking are 
$\Sigma_1^1$-hard and in $\Sigma_2^2$ (for both semantics).

Thus, there are several gaps in these results, i.e., in the complexity of finite-state satisfiability and model-checking when second-order quantification is restricted to least fixed points. 
Similarly, the complexity of satisfiability for this fragment under standard semantics is open.
Finally, the results for guarded quantification rely on using both minimal and maximal sets satisfying guards, i.e., the complexity of the fragments using only one type of \myquot{polarity} are still open. 
In the following, we close all these gaps with tight complexity results.

\section{\texorpdfstring{\boldmath$\sohyltlfp$}{Second-order HyperLTL with Minimality/Maximality Constraints}} 

$\sohyltlfp$ is the fragment of $\sohyltl$ in which second-order quantification is restricted to maximal or minimal sets satisfying a guard formula. 
It is known that all three verification problems we are interested in have the same complexity for $\sohyltlfp$ as for full $\sohyltl$, i.e., they are equivalent to truth in third-order arithmetic~\cite{sohypercomplexity}.
Intuitively, the reason is that one can write a guard over some set that is only satisfied by uncountable models. Using this guard, one can mimic quantification over arbitrary sets of traces, as they are all subsets of an uncountable set.
However, when implementing this idea naively, one ends up with formulas that use both quantification over maximal and minimal sets satisfying a guard.
Here, we show that similar results hold when only using quantification over maximal sets satisfying a guard and when only using quantification over minimal sets satisfying a guard. 

\subsection{Syntax and Semantics}
\label{subsec_mmsands}

The formulas of $\sohyltlfp$ are given by the grammar
\begin{align*}
\phi  {}& \cceq {} \exists (X,\smalar,\phi).\ \phi \mid \forall (X,\smalar,\phi).\ \phi \mid \exists \pi \in X.\ \phi \mid \forall \pi \in X.\ \phi \mid \psi\\
\psi {}&  \cceq {} \prop_\pi \mid \neg \psi \mid \psi \vee \psi \mid \X \psi \mid \psi \U \psi    
\end{align*}
where $\prop$ ranges over $\ap$, $\pi$ ranges over~$\fovar$, $X$ ranges over $\sovar$, and $\smalar \in \set{\smallest, \largest}$, i.e., the only modification concerns the syntax of second-order quantification. 

We consider two fragments of $\sohyltl$ obtained by only allowing quantification over maximal sets (minimal sets, respectively):
\begin{itemize}
    \item $\sohyltlfpmax$ is the fragment of $\sohyltlfp$ obtained by disallowing second-order quantifiers of the form~$\exists (X,\smallest,\phi) $ and $ \forall (X,\smallest,\phi)$.
    \item $\sohyltlfpmin$ is the fragment of $\sohyltlfp$ obtained by disallowing second-order quantifiers of the form~$\exists (X,\largest,\phi) $ and $ \forall (X,\largest,\phi)$.
\end{itemize}

The semantics of $\sohyltlfp$ is equal to that of $\sohyltl$ but for the second-order quantifiers, for which we define (for $\smalar \in \set{\smallest,\largest}$):
\begin{itemize}
    \item $\Pi \models \exists (X,\smalar,\phi_1).\ \phi_2$ if there exists a set~$T \in \solutions(\Pi, (X,\smalar,\phi_1))$ such that $\Pi[X\mapsto T] \models \phi_2$.
    \item $\Pi \models \forall (X,\smalar,\phi_1).\ \phi_2$ if for all sets~$T \in \solutions(\Pi, (X,\smalar,\phi_1))$ we have $\Pi[X\mapsto T] \models \phi_2$.
\end{itemize}
Here, $\solutions(\Pi, (X,\smalar,\phi_1))$ is the set of all minimal/maximal models of the formula~$\phi_1$, which is defined as follows: 
\begin{align*}
    {}&{}\!\!\!\solutions(\Pi, (X,\smallest,\phi_1))  = \set{T \subseteq (\pow{\ap})^\omega \mid \Pi[X\mapsto T] \models \phi_1\text{ and } 
    \Pi[X\mapsto T'] \not\models \phi_1 \text{ for all } T' \subsetneq T }\\
    {}&{}\!\!\!\solutions(\Pi, (X,\largest,\phi_1))  = \set{T \subseteq (\pow{\ap})^\omega \mid \Pi[X\mapsto T] \models \phi_1 \text{ and }
     \Pi[X\mapsto T'] \not\models \phi_1 \text{ for all } T' \supsetneq T}
\end{align*}
Note that $\solutions(\Pi, (X,\smalar,\phi_1))$ may be empty, may be a singleton, or may contain multiple sets, which then are pairwise incomparable.

Let us also define closed-world semantics for $\sohyltlfp$. Here, we again disallow the use of the variable~$\univar$ and change the semantics of set quantification to 
\begin{itemize}
    \item $\Pi \models_\cw \exists (X,\smalar,\phi_1).\ \phi_2$ if there exists a set~$T \in \solutions_\cw(\Pi, (X,\smalar,\phi_1))$ such that $\Pi[X\mapsto T] \models_\cw \phi_2$, and 
    \item $\Pi \models_\cw \forall (X,\smalar,\phi_1).\ \phi_2$ if for all sets~$T \in \solutions_\cw(\Pi, (X,\smalar,\phi_1))$ we have $\Pi[X\mapsto T] \models_\cw \phi_2$,
\end{itemize}
where $\solutions_\cw(\Pi, (X,\smallest,\phi_1))$ and $\solutions_\cw(\Pi, (X,\largest,\phi_1))$ are  defined as follows:
\begin{align*}
    \solutions_\cw(\Pi, (X,\smallest,\phi_1))  = \set{T \subseteq \Pi(\unidisvar) \mid{}&{} \Pi[X\mapsto T] \models_\cw \phi_1\\ 
    {}&{}\text{ and } \Pi[X\mapsto T'] \not\models_\cw \phi_1\text{ for all }T' \subsetneq T} \\
     \solutions_\cw(\Pi, (X,\largest,\phi_1))= \set{T \subseteq \Pi(\unidisvar) \mid{}&{}  \Pi[X\mapsto T] \models_\cw \phi_1\\ 
    {}&{}\text{ and }  \Pi[X\mapsto T'] \not\models_\cw \phi_1 \text{ for all } \Pi(\unidisvar) \supseteq  T' \supsetneq T}.
\end{align*}
Note that $\solutions_\cw(\Pi, (X,\smalar,\phi_1))$ may still be empty, may be a singleton, or may contain multiple sets, but all sets in it are now incomparable subsets of $\Pi(\unidisvar)$.

A $\sohyltlfp$ formula is a sentence if it does not have any free variables except for $\univar$ and $ \unidisvar$ (also in the guards). 
Models are defined as for $\sohyltl$.

\begin{proposition}[Proposition~1 of \cite{DBLP:conf/cav/BeutnerFFM23}]
\label{prop_fp2classical}
Every $\sohyltlfp$ sentence~$\phi$ can be translated in polynomial time (in $\size{\varphi}$) into a $\sohyltl$ sentence~$\phi'$ such that for all sets~$T$ of traces we have that $T \models \phi$  if and only if  $T \models \phi'$.\footnote{The polynomial-time claim is not made in \cite{DBLP:conf/cav/BeutnerFFM23}, but follows from the construction when using appropriate data structures for formulas.} 
\end{proposition}

The same claim is also true for closed-world semantics, using the same proof.

\begin{remark}
\label{remark_fp2classical_cw}
Every $\sohyltlfp$ sentence~$\phi$ can be translated in polynomial time (in $\size{\varphi}$) into a $\sohyltl$ sentence~$\phi'$ such that for all sets~$T$ of traces we have that $T \models_\cw \phi$  if and only if  $T \models_\cw \phi'$.
\end{remark}

Thus, every complexity upper bound for $\sohyltl$ also holds for $\sohyltlfp$ and every lower bound for $\sohyltlfp$ also holds for $\sohyltl$.

\subsection{Satisfiability, Finite-State Satisfiability, and Model-Checking} 

In this subsection, we settle the complexity of satisfiability, finite-state satisfiability, and model-checking for the fragments~$\sohyltlfpmax$ and $\sohyltlfpmin$ using only second-order quantification over maximal respectively minimal sets satisfying a given guard.
We will show, for both semantics, that all three problems have the same complexity as the corresponding problems for full $\sohyltl$ and for $\sohyltlfp$ (which allows quantification over maximal \emph{and} minimal sets satisfying a given guard), i.e., they are equivalent to truth in third-order arithmetic.
Thus, we show that even restricting to one type of monotonicity does not lower the complexity. 

As just explained, the upper bounds already hold for full $\sohyltl$, hence the remainder of this subsection is concerned with lower bounds. 
We show that one can translate each $\sohyltl$ sentence~$\phi$ (over $\ap$) into a $\sohyltlfpmax$ sentence~$\phi^\largest$ and into a $\sohyltlfpmin$ sentence~$\phi^\smallest$ (both over some $\ap' \supsetneq \ap$) and each $T \subseteq (\pow{\ap})^\omega$ into a $T' \subseteq (\pow{\ap'})^\omega$ such that $T\models_\cw \phi$ if and only if $T'\models_\cw \phi^\largest$ and $T\models_\cw \phi$ if and only if $T'\models_\cw \phi^\smallest$.
As closed-world semantics can be reduced to standard semantics, this suffices to prove the result for both semantics.

Intuitively, $\phi^\largest$ and $\phi^\smallest$ mimic the quantification over arbitrary sets in $\phi$ by quantification over maximal and minimal sets that satisfy a guard~$\phi_1$ that is only satisfied by uncountable sets.
It is known that such formulas~$\phi_1$ can be written in $\sohyltl$~\cite{sohypercomplexity}. 
Here, we show that similar formulas can also be written in $\sohyltlfpmax$ and $\sohyltlfpmin$.
With these formulas as guards (which use fresh propositions in $\ap' \setminus \ap$), we mimic arbitrary set quantification via quantification of sets of traces over $\ap'$ that are uncountable (enforced by the guards) and then consider their $\ap$-projections.
This approach works for all sets but the empty set, as projecting an uncountable set cannot result in the empty set. 
For this reason, we additionally mark some traces in the uncountable set and only project the marked ones, but discard the unmarked ones.
Thus, by marking no trace, the projection is the empty set.

\begin{theorem}
\label{thm_hyltlmm_complexity}
\hfill
\begin{enumerate}
    \item $\sohyltlfpmax$ satisfiability, finite-state satisfiability, and model-checking (both under standard semantics and under closed world semantics) are polynomial-time equivalent to truth in third-order arithmetic. The lower bounds for standard semantics already hold for $\univar$-free sentences.

    \item $\sohyltlfpmin$ satisfiability, finite-state satisfiability, and model-checking (both under standard semantics and under closed world semantics) are polynomial-time equivalent to truth in third-order arithmetic. The lower bounds for standard semantics already hold for $\univar$-free sentences.

\end{enumerate}
\end{theorem}

Recall that we only need to prove the lower bounds, as the upper bounds already hold for full $\sohyltl$.
Furthermore, as closed-world semantics can be reduced to standard semantics, we only need to prove the lower bounds for closed-world semantics:
Frenkel and Zimmermann showed  that every $\sohyltl$ sentence~$\phi$ can be translated in polynomial time (in $\size{\varphi}$) into a $\sohyltl$ sentence~$\phi'$ such that for all sets~$T$ of traces we have that $T \models_\cw \phi$  if and only if  $T \models \phi'$ (under standard semantics)~\cite{sohypercomplexity}.
Note that in this translation, a $\univar$-free $\phi$ is mapped to a $\univar$-free $\phi'$.

We begin by constructing the desired guards that have only uncountable models.
As a first step, we modify a construction due to Hadar and Zimmermann that yielded a $\sohyltl$ formula that has only uncountable models: we show that $\sohyltlfpmax$ and $\sohyltlfpmin$ also have formulas that require the interpretation of a free second-order variable to be uncountable.
To this end, fix $\ap_\allsets = \set{\posprop, \negprop, \setprop,\inprop}$ and consider the language
\begin{align*}
T_\allsets ={}&{} \set{ \set{\posprop}^\omega,\set{\negprop}^\omega} \cup\\
&{}\set{\set{\posprop}^n \set{\inprop,\posprop} \set{\posprop}^\omega \mid n\in\nats} \cup\\
&{}\set{\set{\negprop}^n \set{\inprop,\negprop} \set{\negprop}^\omega \mid n\in\nats}{} \cup\\
&{}\set{ (t(0) \cup \set{\setprop})(t(1) \cup \set{\setprop})(t(2) \cup \set{\setprop}) \cdots \mid t \in (\pow{\set{\inprop}})^\omega},
\end{align*}
which is an uncountable subset of $(\pow{\ap_\allsets})^\omega$.
Figure~\ref{fig_tallsets} depicts a transition system~$\tsys_\allsets$ satisfying $\traces(\tsys_\allsets) = T_\allsets$.

\begin{figure}[h]
    \centering

        \begin{tikzpicture}[thick]
        \def\d{1.25}
            \node[plainnode] (1p) at (0,0) {\small$\set{\posprop}$};
            \node[plainnode] (2p) at (3,0) {\small$\set{\inprop,\posprop}$};
            \node[plainnode] (3p) at (6,0) {\small$\set{\posprop}$};

            \path[->] 
            (-.85,0) edge (1p)
            (1p) edge[loop above] ()
            (1p) edge (2p)
            (3,.85) edge (2p)
            (2p) edge (3p) 
            (3p) edge[loop above] ()
            ;

            \node[plainnode] (1n) at (0,-\d) {\small$\set{\negprop}$};
            \node[plainnode] (2n) at (3,-\d) {\small$\set{\inprop,\negprop}$};
            \node[plainnode] (3n) at (6,-\d) {\small$\set{\negprop}$};

            \path[->] 
            (-.85,-\d) edge (1n)
            (1n) edge[loop below] ()
            (1n) edge (2n)
            (3,-\d-.85) edge (2n)
            (2n) edge (3n) 
            (3n) edge[loop below] ()
            ;

            \node[plainnode] (1s) at (9,0) {\small$\set{\inprop,\setprop}$};
            \node[plainnode] (2s) at (9,-\d) {\small$\set{\setprop}$};

            \path[->] 
            (8.15,-0) edge (1s)
            (8.15,-\d) edge (2s)
            (1s) edge[loop above] ()
            (1s) edge[bend left] (2s)
            (2s) edge[bend left] (1s)
            (2s) edge[loop below] ()
            ;

        \end{tikzpicture}
    
    \caption{The transition system $\tsys_\allsets$ with $\traces(\tsys_\allsets) = T_\allsets$.}
    \label{fig_tallsets}
\end{figure}

\begin{lemma}\hfill
\label{lemma_uncountmodels}
\begin{enumerate}
    \item\label{lemma_uncountmodels_largest}
    There exists a $\sohyltlfpmax$ formula~$\phi_\allsets^\largest$ over $\ap_\allsets$ with a single free (second-order) variable~$Z$ such that $\Pi \models_\cw \varphi_\allsets^\largest$ if and only if the $\ap_\allsets$-projection of $\Pi(Z)$ is $T_\allsets$.

    \item\label{lemma_uncountmodels_smallest}
    There exists a $\sohyltlfpmin$ formula~$\phi_\allsets^\smallest$ over $\ap_\allsets$ with a single free (second-order) variable~$Z$ such that $\Pi \models_\cw \varphi_\allsets^\smallest$ if and only if the $\ap_\allsets$-projection of $\Pi(Z)$ is $T_\allsets$.
\end{enumerate}
\end{lemma}

\begin{proof}
\ref{lemma_uncountmodels_largest}.)
Consider $\phi_\allsets^\largest = \phi_0 \wedge \cdots \wedge \phi_4$ where
\begin{itemize}
    \item $\phi_0 = \forall \pi \in Z.\ \bigvee_{\prop \in \set{\posprop, \negprop,\setprop}} \G (\prop_\pi \wedge \bigwedge_{\prop' \in \set{\posprop, \negprop,\setprop} \setminus \set{\prop}} \neg \prop'_\pi )$ expresses that on each trace in the $\ap_\allsets$-projection of $\Pi(Z)$ exactly one of the propositions in $\set{\posprop, \negprop,\setprop}$ holds at each position and the other two at none. 
    In the following, we speak therefore about type~$\prop$ traces for $\prop \in \set{\posprop, \negprop,\setprop}$,

    \item $\phi_1 = \forall \pi \in Z.\  (\posprop_\pi \vee \negprop_\pi) \rightarrow ( (\G \neg \inprop_\pi) \vee (\neg \inprop_\pi \U ( \inprop_\pi \wedge \X \G \neg \inprop_\pi )))$
    expresses that $\inprop$ appears at most once on each type~$\prop$ trace in the $\ap_\allsets$-projection of $\Pi(Z)$, for both $\prop \in \set{\posprop,\negprop}$,

    \item $\phi_2 = \bigwedge_{\prop \in \set{\posprop,\negprop}} (\exists \pi \in Z.\ \prop_\pi\wedge \G\neg \inprop_\pi) \wedge (\exists\pi \in Z.\ \prop_\pi \wedge \inprop_\pi)$ expresses that the type~$\prop$ traces~$\set{\prop}^\omega$ and $\set{\prop}^0 \set{\inprop,\prop} \set{\prop}^\omega$ are in the $\ap_\allsets$-projection of $\Pi(Z)$, for both $\prop \in \set{\posprop,\negprop}$, and

    \item $\phi_3 = \bigwedge_{\prop \in \set{\posprop,\negprop}}\forall \pi \in Z.\ \exists\pi'\in Z.\ (\prop_\pi \wedge \F \inprop_\pi) \rightarrow ( \prop_{\pi'} \wedge  \F (\inprop_\pi \wedge \X \inprop_{\pi'}))$ expresses that for each type~$\prop$ trace of the form~$\set{\prop}^n \set{\inprop,\prop} \set{\prop}^\omega$ in the $\ap_\allsets$-projection of $\Pi(Z)$, the trace~$\set{\prop}^{n+1} \set{\inprop,\prop} \set{\prop}^\omega$ is also in the $\ap_\allsets$-projection of $\Pi(Z)$, for both $\prop \in \set{\posprop,\negprop}$.
    
\end{itemize}
The conjunction of these first four formulas requires that the $\ap_\allsets$-projection of $\Pi(Z) $ (and thus, under closed-world semantics, also the model) contains at least the traces in
\begin{equation}
    \label{eq_nats}
\set{ \set{\posprop}^\omega,\set{\negprop}^\omega} \cup \set{\set{\posprop}^n \set{\inprop,\posprop} \set{\posprop}^\omega \mid n\in\nats} \cup \set{\set{\negprop}^n \set{\inprop,\negprop} \set{\negprop}^\omega \mid n\in\nats}.
\end{equation}

We say that a set~$T$ of traces is contradiction-free if there is no $n \in \nats$ such that $\set{\posprop}^n \set{\inprop,\posprop} \set{\posprop}^\omega$ and $\set{\negprop}^n \set{\inprop,\negprop} \set{\negprop}^\omega$ are in the $\ap_\allsets$-projection of $T$.
A trace~$t$ is consistent with a contradiction-free~$T$ if the following two conditions are satisfied:
\begin{description}
    \item[(C1)] If $\set{\posprop}^n \set{\inprop,\posprop} \set{\posprop}^\omega $ is in the $\ap_\allsets$-projection of $T$ then $\inprop\in t(n)$.
    \item[(C2)] If $\set{\negprop}^n \set{\inprop,\negprop} \set{\negprop}^\omega $ is in the $\ap_\allsets$-projection of $T$ then $\inprop\notin t(n)$.
\end{description}
Note that $T$ does not necessarily specify the truth value of $\inprop$ in every position of $t$, i.e., in those positions~$n\in\nats$ where neither $\set{\posprop}^n \set{\inprop,\posprop} \set{\posprop}^\omega$ nor $\set{\negprop}^n \set{\inprop,\negprop} \set{\negprop}^\omega$ are in $T$.
Nevertheless, due to (\ref{eq_nats}), for every trace~$t$ over $\set{\inprop}$ there exists a contradiction-free subset~$T$ of the $\ap_\allsets$-projection of $\Pi(Z)$ such that the $\set{\inprop}$-projection of every trace~$t'$ that is consistent with $T$ is equal to $t$. 
Here, we can, w.l.o.g., restrict ourselves to maximal contradiction-free sets, i.e., sets that stop being contradiction-free if more traces are added.
Thus, each of the uncountably many traces over $\set{\inprop}$ is induced by some maximal contradiction-free subset of the $\ap_\allsets$-projection of $\Pi(Z)$.
\begin{itemize}
\item Hence, we define $\phi_4$ as the formula
\begin{align*}
\forall {}&{}(X, \largest, \overbrace{ \forall \pi \in X.\ \forall \pi' \in X.\ (\posprop_\pi \wedge \negprop_{\pi'}) \rightarrow \neg \F(\inprop_{\pi} \wedge \inprop_{\pi'}) }^{\text{$X$ is contradiction-free}}).\\
&{}\exists \pi'' \in Z.\ \forall \pi''' \in X.\ \setprop_{\pi''} \wedge  \underbrace{((\posprop_{\pi'''}\wedge \F\inprop_{\pi'''}) \rightarrow \F(\inprop_{\pi'''} \wedge \inprop_{\pi''}))}_{\text{(C1)}}
\wedge 
\underbrace{((\negprop_{\pi'''} \wedge \F\inprop_{\pi'''}) \rightarrow \F(\inprop_{\pi'''} \wedge \neg\inprop_{\pi''}))}_{\text{(C2)}},
\end{align*} 
expressing that for every maximal contradiction-free set of traces~$T$, there exists a type~$\setprop$ trace~$t''$ in the $\ap_\allsets$-projection of $\Pi(Z)$ that is consistent with $T$.
\end{itemize}

Thus, if the $\ap_\allsets$-projection of $\Pi(Z)$ is $T_\allsets$, then $\Pi\models_\cw\phi_\allsets^\largest$.
Dually, we can conclude that $T_\allsets$ must be a subset of the $\ap_\allsets$-projection of $\Pi(Z)$ whenever $\Pi\models_\cw \phi_\allsets^\largest$, and that the $\ap_\allsets$-projection of $\Pi(Z)$ cannot contain other traces (due to $\phi_0$ and $\phi_1$).
Hence, $\Pi \models_\cw \varphi^\largest$ if and only if the $\ap_\allsets$-projection of $\Pi(Z)$ is $T_\allsets$.

\ref{lemma_uncountmodels_smallest}.)
Here, we will follow a similar approach, but have to overcome one obstacle: there exists a unique minimal contradiction-free set, i.e., the empty set.
Hence, we cannot naively replace quantification over maximal contradiction-free sets in $\phi_4$ above by quantification over minimal contradiction-free sets. 
Instead, we will quantify over minimal contradiction-free sets that have, for each $n$, either the trace~$\set{\posprop}^n \set{\inprop,\posprop} \set{\posprop}^\omega$ or the trace~$\set{\negprop}^n \set{\inprop,\negprop} \set{\negprop}^\omega$ in their $\ap_\allsets$-projection. 
The minimal sets satisfying this constraint are still \emph{rich} enough to enforce every possible trace over $\set{\inprop}$.

Formally, we replace $\phi_4$ by $\phi_4'$ defined as 
\begin{align*}
\forall {}&{}(X, \smallest, { \phi_\complete \wedge \forall \pi \in X.\ \forall \pi' \in X.\ (\posprop_\pi \wedge \negprop_{\pi'}) \rightarrow \neg \F(\inprop_{\pi} \wedge \inprop_{\pi'}) }) .\\
&{}\exists \pi'' \in Z.\ \forall \pi''' \in X.\ \setprop_{\pi''} \wedge  {(\posprop_{\pi'''} \rightarrow \F(\inprop_{\pi'''} \wedge \inprop_{\pi''}))}
\wedge 
{(\negprop_{\pi'''} \rightarrow \F(\inprop_{\pi'''} \wedge \neg\inprop_{\pi''}))},
\end{align*} 
i.e., the only changes are the change of the polarity from $\largest$ to $\smallest$ and the addition of $\phi_\complete$ in the guard, which is the formula
\[
\exists \pi \in X.\ (\inprop_\pi \wedge (\posprop_\pi \vee \negprop_\pi)) \wedge \forall \pi \in X.\ \exists \pi' \in X.\ (\posprop_\pi \vee \negprop_\pi) \rightarrow ((\posprop_{\pi'} \vee \negprop_{\pi'}) \wedge \F( \inprop_\pi \wedge \X\inprop_{\pi'}))
.
\]
This ensures that $\phi_\allsets^\smallest = \phi_0 \wedge \cdots \wedge \phi_3 \wedge \phi_4'$ has the desired properties.
\end{proof}

Now, we can begin the translation of full $\sohyltl$ into $\sohyltlfpmax$ and $\sohyltlfpmin$.
Let us fix a $\sohyltl$ sentence~$\varphi$ over a set~$\ap$ of propositions that is, without loss of generality, disjoint from $\ap_\allsets$.
Hence, satisfaction of $\phi$ only depends on the projection of traces to $\ap$, i.e., if $T_0$ and $T_1$ have the same $\ap$-projection, then $T_0 \models_\cw \phi$ if and only if $T_1 \models_\cw\phi$. 
We assume without loss of generality that each variable is quantified at most once in $\phi$ and that $\univar$ and $\unidisvar$ are not bound by a quantifier in $\phi$, which can always be achieved by renaming variables. 
Let $X_0, \ldots, X_{k-1}$ be the second-order variables quantified in $\phi$. 
To simplify our notation, we define $[k] = \set{0,1,\ldots, k-1}$.
For each $i \in [k]$, we introduce a fresh proposition~$\marker_{i}$ so that we can define $\ap'$ as the pairwise disjoint union of $\ap$, $\set{\marker_{i} \mid i \in [k]}$, and $\ap_\allsets$.

Let $i \in [k]$ and consider the formula
\begin{align*}
\phi_\parti^{i} ={}&{} \forall \pi \in X_i.\ ((\G ({\marker_{i}})_\pi) \vee (\G \neg ({\marker_{i}})_\pi) \wedge \bigwedge_{i' \in [k]\setminus\set{i}}\G\neg(\marker_{i'})_\pi ) \wedge \\
{}&{}\forall \pi \in X_i.\ \forall \pi' \in X_i.\ (\equals{\pi}{\pi'}{\ap_\allsets }) \rightarrow (\equals{\pi}{\pi'}{\ap'}).   
\end{align*}
Note that $\phi_\parti^{i}$ has a single free variable, i.e., $X_i$, and that it is both a formula of $\sohyltlfpmin$ and $\sohyltlfpmax$.
Intuitively, $\phi_\parti^{i}$ expresses that each trace in $\Pi(X_i)$ is either marked by $\marker_{i}$ (if ${\marker_{i}}$ holds at every position) or it is not marked (if ${\marker_{i}}$ holds at no position), it is not marked by any other $\marker_{i'}$, and that there may \emph{not} be two distinct traces~$t \neq t'$ in $\Pi(X_i)$ that have the same $\ap_\allsets$-projection.
The former condition means that $\Pi(X_i)$ is partitioned into two (possibly empty) parts, the subset of marked traces and the set of unmarked traces; the latter condition implies that each trace in $\pi(X_i)$ is uniquely identified by its $\ap_\allsets$-projection.

Fix some $i \in [k]$ and some $\smalar \in \set{\largest,\smallest}$, and define 
\[\phi_\guard^{i,\smalar} = \phi_\allsets^{\smalar}[Z/X_i] \wedge \phi_\parti^{i}, \]
where $\phi_\allsets^{\smalar}[Z/X_i]$ is the formula obtained from $\phi_\allsets^{\smalar}$ by replacing each occurrence of $Z$ by $X_i$.
The only free variable of $\phi_\guard^{i,\smalar}$ is $X_i$.

\begin{lemma}
\label{lemma_guardcorrectness}
Let $\Pi_0'$ and $\Pi_1'$ be two variable assignments with $\Pi_0'(X_i) \subseteq (\pow{\ap'})^\omega$ and  $\Pi_1'(X_i) \subseteq (\pow{\ap'})^\omega$.
\begin{enumerate}
    
    \item\label{lemma_guardcorrectness_largest}
    If $\Pi_0' \models_\cw \phi_\guard^{i,\largest}$ and $\Pi_1'(X_i) \supsetneq \Pi_0'(X_i)$, then $\Pi_1' \not\models_\cw \phi_\guard^{i,\largest}$.

    \item\label{lemma_guardcorrectness_smallest}
    If $\Pi_0' \models_\cw \phi_\guard^{i,\smallest}$ and $\Pi_1'(X_i) \subsetneq \Pi_0'(X_i)$, then $\Pi_1' \not\models_\cw \phi_\guard^{i,\smallest}$.

\end{enumerate}
\end{lemma}

\begin{proof}
\ref{lemma_guardcorrectness_largest}.) 
Towards a contradiction, assume we have $\Pi_0' \models_\cw \phi_\guard^{i,\largest}$ and $\Pi_1'(X_i) \supsetneq \Pi_0'(X_i)$, but also $\Pi_1' \models_\cw \phi_\guard^{i,\largest}$.
Then, there exists a trace~$t_1 \in \Pi_1'(X_i) \setminus \Pi_0'(X_i)$. 
As $\Pi_0'\models_\cw \phi_\allsets^\largest$, which was constructed to satisfy Lemma~\ref{lemma_uncountmodels}.\ref{lemma_uncountmodels_largest}, there exists a trace~$t_0 \in \Pi_0'(X_i)$ that has the same $\ap_\allsets$-projection as $t_1$. 
Hence, $\varphi_\parti^{i}$ implies that $t_0$ and $t_1$ have the same $\ap'$-projection, i.e., they are the same trace (here we use $\Pi_0'(X_i) \subseteq (\pow{\ap'})^\omega$ and  $\Pi_1'(X_i) \subseteq (\pow{\ap'})^\omega$).
Thus, $t_1 = t_0$ is in $\Pi_0'(X_i)$, i.e., we have derived a contradiction.

\ref{lemma_guardcorrectness_smallest}.)
The argument here is similar, we just have to swap the roles of the two sets~$\Pi_0'(X_i)$ and 
$\Pi_1'(X_i)$.

Towards a contradiction, assume we have $\Pi_0' \models_\cw \phi_\guard^{i,\smallest}$ and $\Pi_1'(X_i) \subsetneq \Pi_0'(X_i)$, but also $\Pi_1' \models_\cw \phi_\guard^{i,\smallest}$.
Then, there exists a trace~$t_0 \in \Pi_0'(X_i) \setminus \Pi_1'(X_i)$. 
As $\Pi_1'\models_\cw \phi_\allsets^\smallest$, which was constructed to satisfy Lemma~\ref{lemma_uncountmodels}.\ref{lemma_uncountmodels_smallest}, there exists a trace~$t_1 \in \Pi_1'(X_i)$ that has the same $\ap_\allsets$-projection as $t_0$. 
Hence, $\varphi_\parti^{i}$ implies that $t_1$ and $t_0$ have the same $\ap'$-projection, i.e., they are the same trace.
Thus, $t_0 = t_1$ is in $\Pi_1'(X_i)$, i.e., we have derived a contradiction.
\end{proof}

Recall that our goal is to show that quantification over minimal/maximal subsets of $(\pow{\ap'})^\omega$ satisfying $\phi_\guard^{i,\smalar}$ mimics quantification over subsets of $(\pow{\ap})^\omega$.
Now, we can make this statement more formal.
Let $T' \subseteq (\pow{\ap'})^\omega$. We define
\[
\enc_i(T') = \set{
t \in (\pow{\ap})^\omega \mid t \text{ is the $\ap$-projection of some $t' \in T'$ whose $\set{\marker_{i}}$-projection is $\set{\marker_{i}}^\omega$}
}.
\]
Now, every $\enc_i(T')$ is, by definition, a subset of $(\pow{\ap})^\omega$. 
Our next result shows that, conversely, every subset of $(\pow{\ap})^\omega$ 
can be obtained as an encoding of some $T'$ that additionally satisfies the guard formulas.

\begin{lemma}
\label{lemma_encodingrichness}
Let $T \subseteq (\pow{\ap})^\omega$, $i \in [k]$, and $\smalar \in \set{\largest,\smallest}$.
there exists a $T' \subseteq (\pow{\ap'})^\omega$ such that
\begin{itemize}
    \item $T = \enc_i(T')$, and 
    \item for all $\Pi$ with $\Pi(X_i) = T'$, we have $\Pi \models_\cw \phi_\guard^{i,\smalar}$.
\end{itemize}
\end{lemma}

\begin{proof}
Fix a bijection~$f \colon (\pow{\ap})^\omega \rightarrow T_\allsets$, which can, e.g., be obtained by applying the Schröder-Bernstein theorem.
Now, we define 
\[T' = \set{t\merge f(t) \merge\set{\marker_{i}}^\omega \mid t \in T} \cup \set{t\merge f(t) \mid t \in (\pow{\ap})^\omega \setminus T}. \]
By definition, we have $\enc_i(T') = T$, satisfying the first requirement on $T'$.
Furthermore, the $\ap_\allsets$-projection of $T'$ is $T_\allsets$, each trace in $T'$ is either marked by $\marker_{i}$ at every position or at none, the other markers do not appear in traces in $T'$, and, due to $f$ being a bijection, there are no two traces in $T'$ with the same $\ap_\allsets$-projection.
Hence, due to Lemma~\ref{lemma_uncountmodels}, $T'$ does indeed satisfy $\phi_\guard^{i,\smalar}$.
\end{proof}

Now, we explain how to mimic quantification over arbitrary sets via quantification over maximal or minimal sets satisfying the guards we have constructed above.
Let $\smalar \in \set{\smallest, \largest}$ and let $\varphi^{\smalar}$ be the $\sohyltlfpsmalar$-sentence obtained from $\phi$ by inductively replacing 
\begin{itemize}
    \item each $\exists X_i.\phi'$ by $\exists (X_i, \smalar, \varphi_\guard^{i,\smalar})$, 
    \item each $\forall X_i.\phi'$ by $\forall (X_i, \smalar, \varphi_\guard^{i,\smalar})$, 
    \item each $\exists \pi \in X_i.\ \phi'$ by $\exists \pi \in X_i.\ (\marker_{i})_{\pi} \wedge \psi'$, and 
    \item each $\forall \pi \in X_i.\ \phi'$ by $\forall \pi \in X_i.\ (\marker_{i})_{\pi} \rightarrow \psi'$.
\end{itemize}

We show that $\phi$ and $\phi^{\smalar}$ are in some sense equivalent. 
Obviously, they are not equivalent in the sense that they have the same models as $\phi^{\smalar}$ uses in general the additional propositions in $\ap' \setminus \ap$, which are not used by $\phi$.
But, by \emph{extending} a model of $\phi$ we obtain a model of $\phi^{\smalar}$.
Similarly, by \emph{ignoring} the additional propositions (i.e., the inverse of the extension) in a model of $\phi^{\smalar}$, we obtain a model of $\phi$.

As the model is captured by the assignment to the variable~$\unidisvar$, this means the interpretation of $\unidisvar$ when evaluating $\phi$ and the the interpretation of $\unidisvar$ when evaluating $\phi^{\smalar}$ have to satisfy the extension property to be defined.
However, we show correctness of our translation by induction.
Hence, we need to strengthen the induction hypothesis and require that the variable assignments for $\varphi$ and $\phi^{\smalar}$ satisfy the extension property for all free variables.

Formally, let $\Pi$ and $\Pi'$ be two variable assignments such that 
\begin{itemize}
    \item $\Pi(\pi) \in (\pow{\ap})^\omega$ and $\Pi(X) \subseteq (\pow{\ap})^\omega$ for all $\pi$ and $X$ in the domain of $\Pi$, and 
    \item $\Pi'(\pi) \in (\pow{\ap'})^\omega$ and $\Pi'(X) \subseteq (\pow{\ap'})^\omega$ for all $\pi$ and $X$ in the domain of $\Pi'$.
\end{itemize}
We say that $\Pi'$ extends $\Pi$ if they have the same domain (which must contain $\unidisvar$, but not $\univar$ as this variable may not be used under closed-world semantics) and we have that
\begin{itemize}
    \item for all $\pi$ in the domain, the $\ap$-projection of $\Pi'(\pi)$ is $\Pi(\pi)$, 
    \item for all $X_i$ for $i \in [k]$, in the domain, $\Pi(X_i) = \enc_i(\Pi'(X_i))$, and
    \item $\Pi'(\unidisvar) = \extend(\Pi(\unidisvar))$ where
    \[\extend(T) = \set{t \merge t'{} \merge \set{\marker_{i}} \mid t \in T, t' \in T_\allsets, \text{ and } i \in [k]} \cup \set{t \merge t' \mid t \in T \text{ and } t' \in T_\allsets},\]
    which implies that the $\ap$-projection of $\Pi'(\unidisvar)$ is $\Pi(\unidisvar)$.
\end{itemize}

Note that if $T$ is a subset of $\Pi(\unidisvar)$, then there exists a subset~$T'$ of $\Pi'(\unidisvar)$ such that $\enc_i(T') = T$, i.e., $\Pi'(\unidisvar)$ contains enough traces to mimic the quantification of subsets of $\Pi(\unidisvar)$ using our encoding.
Furthermore, if $T' \subseteq (\pow{\ap'})^\omega$  has the form~$\extend(T)$ for some $T \subseteq (\pow{\ap})^\omega$, then this $T$ is unique.

The following lemma states that our translation of $\sohyltl$ into $\sohyltlfpmax$ and $\sohyltlfpmin$ is correct. 

\begin{lemma}
\label{lemma_mmcorrectness}
Let $\Pi'$ extend $\Pi$. Then, $\Pi \models_\cw \phi$ if and only if $\Pi' \models_\cw \varphi^{\smalar}$.
\end{lemma}

\begin{proof}
By induction over the subformulas~$\psi$ of $\phi$.

First, let us consider the case of atomic propositions, i.e., $\psi = \prop_\pi$ for some $\prop \in \ap$ and some $\pi$ in the domains of $\Pi$ and $\Pi'$. We have 
\[
\Pi\models_\cw\prop_\pi \Leftrightarrow 
\prop \in \Pi(\pi)(0) \Leftrightarrow 
\prop \in \Pi'(\pi)(0) \Leftrightarrow 
\Pi'\models_\cw\prop_\pi,
\]
as $\Pi(\pi)$ and $\Pi'(\pi)$ have the same $\ap$-projection due to $\Pi'$ extending $\Pi$.

The cases of Boolean and temporal operators are straightforward applications of the induction hypothesis. So, it only remains to consider the quantifiers. 

Let $\psi = \exists X_i.\ \psi_0$ for some $i \in [k]$. Then, we have $\varphi^{\smalar} = \exists (X_i, \smalar, \varphi_\guard^{i,\smalar}).\ \psi_0^{\smalar}$.
By induction hypothesis, we have for all $T \subseteq (\pow{\ap})^\omega$ and $T' \subseteq (\pow{\ap'})^\omega$ with $\enc_i(T') = T$, the equivalence
\[\Pi[X_i \mapsto T] \models_\cw \psi_0 \Leftrightarrow \Pi'[X_i \mapsto T' ] \models_\cw \psi_0^{\smalar}.\]
Now, we have 
\begin{align*}
&{}\Pi\models_\cw \psi\\
 \Leftrightarrow{} &{}\text{there exists a $T \subseteq \Pi(\unidisvar)$ such that } \Pi[X_i \mapsto T] \models_\cw \psi_0 \\
\xLeftrightarrow{\ast}{}&{}\text{there exists a $T' \subseteq \Pi'(\unidisvar)$ such that } \Pi'[X_i \mapsto T'] \models_\cw \psi_0^{\smalar} \text{ and  } T' \in \solutions_\cw(\Pi,( X_i, \smalar, \varphi_\guard^{i,\smalar} ))
\\
\Leftrightarrow{}&{}\Pi'\models_\cw \psi^{\smalar}.
\end{align*}
Here, the equivalence marked with $\ast$ is obtained by applying the induction hypothesis:
\begin{itemize}
    \item For the left-to-right direction, given $T$, we pick $T'$ such that $\enc_i(T') = T$, which is always possible due to Lemma~\ref{lemma_encodingrichness}. Furthermore, $T'$ is in $\solutions_\cw(\Pi,( X_i, \smalar, \varphi_\guard^{i,\smalar} )$, due to Lemma~\ref{lemma_encodingrichness} and Lemma~\ref{lemma_guardcorrectness}.
    \item For the right-to-left direction, given $T'$, we pick $T = \enc_i(T')$.
\end{itemize}
Thus, in both directions, $\Pi'[X_i \mapsto T']$ extends $\Pi[X_i \mapsto T] $, i.e., the induction hypothesis is indeed applicable.
The argument for universal set quantification is dual.

So, it remains to consider trace quantification. 
First, let $\psi= \exists \pi \in X_i.\ \psi_0$ for some $i \in [k]$, which implies
$\psi^{\smalar} = \exists \pi \in X.\ (\marker_{X})_{\pi} \wedge \psi'$.
Now,
\begin{align*}
{}&{}\Pi \models_\cw \psi \\
\Leftrightarrow{}&{}
\text{there exists a $t \in \Pi(X_i)$ such that } \Pi[\pi \mapsto t] \models_\cw \psi_0 \\
\xLeftrightarrow{\ast}{}&{}
\text{there exists a $t' \in \Pi'(X_i)$ such that } \Pi'[\pi \mapsto t'] \models_\cw \psi_0^{\smalar} \\
\Leftrightarrow{}&{}
\Pi' \models_\cw (\marker_{i})_\pi \wedge \psi^{\smalar}.
\end{align*}
Here, the equivalence marked by $\ast$ follows from the induction hypothesis: 
\begin{itemize}
    \item For the left-to-right direction, given $t$, we can pick a $t'$ with the desired properties as $\Pi'$ extends $\Pi$, which implies that $\Pi(X_i) = \enc_i(\Pi'(X_i))$, which in turn implies that $\Pi'(X_i)$ contains a trace~$t'$ marked by $\marker_{i}$ whose $\ap$-projection is $t$.
    \item For the right-to-left direction, given $t'$, we pick $t$ as the $\ap$-projection of $t'$, which is in $\Pi(X_i)$, as $\Pi(X_i) = \enc_i(\Pi'(X_i))$ and $t'$ is marked by $\marker_i$ due to $\Pi' \models_\cw (\marker_{i})_\pi$.
\end{itemize}
Thus, in both directions, $\Pi'[\pi \mapsto t']$ extends $\Pi[\pi \mapsto t] $, i.e., the induction hypothesis is indeed applicable.
The argument for universal trace quantification is again dual. 

So, it only remains to consider trace quantification over $\unidisvar$.
Consider $\psi = \exists \pi \in \unidisvar.\ \psi_0$, which implies $\psi^{\smalar} = \exists \pi \in \unidisvar.\ \psi_0^{\smalar}$.
Then, we have
\begin{align*}
{}&{}\Pi\models_\cw \psi \\
\Leftrightarrow{}&{} \text{there exists a $t \in \Pi(\unidisvar)$ such that $\Pi[\pi \mapsto t] \models_\cw \psi_0$} \\
\xLeftrightarrow{\ast}{}&{} \text{there exists a $t' \in \Pi'(\unidisvar)$ such that $\Pi'[\pi \mapsto t'] \models_\cw \psi_0^{\smalar}$}\\
\Leftrightarrow{}&{} \Pi'\models_\cw \psi^{\smalar}
,    
\end{align*}
where the equivalence marked with $\ast$ follows from the induction hypothesis, which is applicable as the $\ap$-projection of $\Pi'(\unidisvar)$ is equal to $\Pi(\unidisvar)$, which implies that we can choose $t'$ from 
$t$ for the left-to-right direction (and choose $t$ from 
$t'$ for the right-to-left direction) such that $\Pi'[\pi \mapsto t']$ extends $\Pi[\pi \mapsto t]$.
Again, the argument for universal quantification is dual.
\end{proof}

Recall that $\phi$ is satisfied by some set~$T \subseteq (\pow{\ap})^\omega$ of traces (under closed-world semantics) if $\Pi_\emptyset[\unidisvar\mapsto T] \models_\cw \varphi$.
Similarly, $\phi^\smalar$ is satisfied by some set~$T' \subseteq (\pow{\ap'})^\omega$ of traces (under closed-world semantics) if $\Pi_\emptyset[\unidisvar\mapsto T'] \models_\cw \varphi^\smalar$.
The following corollary of Lemma~\ref{lemma_mmcorrectness} holds as $\Pi_\emptyset[\unidisvar\mapsto \extend(T)]$ extends $\Pi_\emptyset[\unidisvar\mapsto T]$.

\begin{corollary}
\label{coro_mmcorrectness}
Let $T\subseteq (\pow{\ap})^\omega$, and $\smalar \in \set{\largest,\smallest}$. Then, $T \models_\cw \phi$ if and only if $\extend(T) \models_\cw \phi^\smalar$.
\end{corollary}

So, to conclude the construction, we need to ensure that models of $\phi^\smalar$ have the form~$\extend(T)$ for some $T \subseteq (\pow{\ap})^\omega$.
For the two satisfiability problems, we do so using a sentence that only has such models while for the model-checking problem, we can directly transform the transition system we are checking so that it satisfies this property.

First, we construct $\phi_\extend^{\smalar}$ for $\smalar \in \set{\largest,\smallest}$ such that $\Pi'\models_\cw \phi_\extend^{\smalar}$ if and only if $\Pi'(\unidisvar) = \extend(T)$ for some $T \subseteq (\pow{\ap})^\omega$:
\begin{align*}
\phi_\extend^{\smalar} ={}&{} (\forall \pi \in \unidisvar.\ (\prop_\pi \wedge \neg \prop_\pi)) \vee \\
{}&{}\phi_\alltraces^{\smalar} \wedge \forall \pi \in \unidisvar.\ \forall \pi' \in \unidisvar.\ \\
{}&{}\quad
\left(
\bigwedge_{i \in [k]} \exists \pi'' \in \unidisvar.\ 
\equals{\pi}{\pi''}{\ap_\allsets} \wedge 
\equals{\pi'}{\pi''}{\ap} \wedge 
\G (\marker_i)_{\pi''} \wedge \bigwedge_{i' \in [k]\setminus\set{i}} \G\neg(\marker_{i'})_{\pi''}
\right) \wedge\\
{}&{}\quad
\left( 
 \exists \pi'' \in \unidisvar.\ 
\equals{\pi}{\pi''}{\ap_\allsets} \wedge 
\equals{\pi'}{\pi''}{\ap} \wedge 
\bigwedge_{i \in [k]} \G\neg(\marker_{i})_{\pi''}
\right)
\end{align*}
Here, the first disjunct is for the special case of the empty $T$ (where $\prop$ is an arbitrary proposition) while the second one expresses intuitively that for each $t \in T_\allsets$ (bound to $\pi$) and each $t'$ in the $\ap$-projection of the interpretation of $\unidisvar$ (bound to $\pi'$), the traces~$t\merge t'$ and the traces~$t\merge t'{}\merge \set{\marker_i}^\omega$, for each $i \in [k]$, are in the interpretation of $\unidisvar$.

Finally, let us consider the transformation of transition systems: Given a transition system~$\tsys = (V, E, I, \lambda)$ we construct a transition system~$\extend(\tsys)$ such that $\traces(\extend(\tsys)) = \extend(\traces(\tsys))$.
Recall the transition system~$\tsys_\allsets$ depicted in Figure~\ref{fig_tallsets}, which satisfies $\traces(\tsys_\allsets) = T_\allsets$. 
Let $\tsys_\allsets = (V_a, E_a, I_a,\lambda_a)$.
We define the transition system~$\extend(\tsys) = (V', E', I', \lambda')$ where 
\begin{itemize}
    \item $V' = V \times V_a \times [k] \times\set{0,1} $,
    \item $E' = \set{((v,v_a,i,b),(v',v_a',i,b)) \mid (v,v') \in E, (v_a, v_a') \in E_a, i \in [k] \text{ and } b \in \set{0,1} }$,
    \item $I' = I \times I_a \times [k] \times\set{0,1} $, and
    \item $\lambda'(v,v_a,i,b) = 
    \begin{cases}
        \lambda(v) \cup \lambda_a(v_a) \cup \set{\marker_i} &\text{ if } b=1,\\
        \lambda(v) \cup \lambda_a(v_a)  &\text{ if } b=0.
    \end{cases}$
\end{itemize}
Note that we indeed have $
\traces(\extend(\tsys)) = \extend(\traces(\tsys))$.
Hence, if $\Pi(\unidisvar) = \traces(\tsys)$ and $\Pi'(\unidisvar) = \traces(\extend(\tsys))$, then $\Pi$ and $\Pi'$ satisfy the requirement spelled out in Corollary~\ref{coro_mmcorrectness}.

\begin{lemma}
\label{lem_so2sofp}
Let $\phi$ be a $\sohyltl$ sentence and $\phi^{\smalar}$ as defined above (for some $\smalar \in \set{\largest,\smallest}$), and let $\tsys$ be a transition system.
\begin{enumerate}
    \item\label{lem_so2sofp_sat} $\phi$ is satisfiable under closed-world semantics if and only if $\phi^{\smalar} \wedge \phi_\extend^{\smalar}$ is satisfiable under closed-world semantics.
    \item\label{lem_so2sofp_fssat} $\phi$ is finite-state satisfiable under closed-world semantics if and only if $\phi^{\smalar} \wedge \phi_\extend^{\smalar}$ is finite-state satisfiable under closed-world semantics. 
    \item\label{lem_so2sofp_mc} $\tsys \models_\cw \phi$ if and only if $\extend(\tsys) \models_\cw \phi^{\smalar}$.  
\end{enumerate}
\end{lemma}

\begin{proof}
\ref{lem_so2sofp_sat}.) 
We have
\begin{align*}
{}&{} \text{$\phi$ is satisfiable under closed-world semantics}\\
\Leftrightarrow {}&{} \text{there exists a $T \subseteq (\pow{\ap})^\omega$ such that $ T \models_\cw \phi$}\\
\xLeftrightarrow{\ast} {}&{} \text{there exists a $T' \subseteq (\pow{\ap'})^\omega$ such that $ T' \models_\cw \phi^{\smalar}\wedge \phi_\extend^{\smalar}$}\\
\Leftrightarrow {}&{} \text{$\phi^{\smalar} \wedge \phi_\extend^{\smalar}$ is satisfiable under closed-world semantics.}
\end{align*}
Here, the equivalence marked with $\ast$ follows from Corollary~\ref{coro_mmcorrectness}:
\begin{itemize}
    \item For the left-to-right direction, we pick $T' = \extend(T)$. 
    \item For the right-to left direction, we pick $T$ to be the unique subset of $(\pow{\ap})^\omega$ such that $T' = \extend(T)$, which is well-defined as $T'$ satisfies $\phi_\extend^{\smalar}$. 
\end{itemize}

\ref{lem_so2sofp_fssat}.)
We have
\begin{align*}
{}&{} \text{$\phi$ is finite-state satisfiable under closed-world semantics}\\
\Leftrightarrow {}&{} \text{there exists a transition system~$\tsys$ over $\pow{\ap}$ such that $\traces(\tsys) \models_\cw \phi$}\\
\xLeftrightarrow{\ast} {}&{} \text{there exists a transition system~$\tsys'$ over $\pow{\ap'}$ such that $\traces(\tsys') \models_\cw \phi^{\smalar}\wedge \phi_\extend^{\smalar}$}\\
\Leftrightarrow {}&{} \text{$\phi^{\smalar} \wedge \phi_\extend^{\smalar}$ is finite-state satisfiable under closed-world semantics.}
\end{align*}
Here, the equivalence marked with $\ast$ follows from Corollary~\ref{coro_mmcorrectness}:
\begin{itemize}
    \item For the left-to-right direction, we pick $\tsys' = \extend(\tsys)$, which satisfies $\traces(\extend(\tsys)) = \extend(\traces(\tsys))$.
    \item For the right-to left direction, we pick $\tsys$ to be the transition system obtained from $\tsys'$ by removing all propositions in $\ap' \setminus \ap$ from the vertex labels.
    This implies that $\traces(\tsys')$ is equal to $\extend(\traces(\tsys))$, as $\traces(\tsys')$ satisfies $\phi_\extend^{\smalar}$.
\end{itemize}

\ref{lem_so2sofp_mc}.) 
Due to Corollary~\ref{coro_mmcorrectness} and $\traces(\extend(\tsys)) = \extend(\traces(\tsys))$, we have $\tsys \models_\cw \phi$ if and only if $\extend(\tsys)\models_\cw \phi^\smalar$.
\end{proof}

Finally, Theorem~\ref{thm_hyltlmm_complexity} is now a direct consequence of Lemma~\ref{lem_so2sofp}, the fact that closed-world semantics can be reduced to standard semantics, and the fact that all three problems for $\sohyltl$ (under closed-world semantics) are equivalent to truth in third-order arithmetic~\cite{sohypercomplexity}.

\section{\texorpdfstring{The Least Fixed Point Fragment of \boldmath$\sohyltlfp$}{Second-order HyperLTL with Least Fixed Points}}
\label{sec_lfp}

When Beutner et al.\ introduced $\sohyltl$ to express, e.g., common knowledge and asynchronous hyperproperties, which are not expressible in $\hyltl$, it turned out that all these examples could be expressed using least fixed points of $\hyltl$-definable operators. 
Thus, they studied the fragment of $\sohyltl$ where second-order quantifiers range only over such least fixed points.
As such fixed points are unique, set quantification therefore degenerates to fixed point computation. 
It is known~\cite{sohypercomplexity} that satisfiability for $\univar$-free sentences in this fragment is much simpler, i.e., $\Sigma_1^1$-complete, than for full $\sohyltl$, which is equivalent to truth in third-order arithmetic.
Furthermore, finite-state satisfiability and model-checking are in $\Sigma_2^2$ and $\Sigma_1^1$-hard, where the lower bounds already hold for $\univar$-free sentences while the upper bounds hold for arbitrary sentences. 

Here, we close these gaps and also settle the complexity of satisfiability for sentences that may contain $\univar$.

\subsection{Syntax and Semantics}
\label{subsec_lfpsands}

A $\sohyltlfp$ sentence using only minimality constraints has the form 
\[
\phi = 
\gamma_1.\ \quant_1(Y_1, \smallest, \phi_1^\con).\ 
\gamma_2.\ \quant_2(Y_2, \smallest, \phi_2^\con).\ 
\ldots 
\gamma_k.\ \quant_k(Y_k, \smallest, \phi_k^\con).\
\gamma_{k+1}.\
\psi
\]
satisfying the following properties: 
\begin{itemize}
    \item Each $\gamma_j$ is a block~$
    \gamma_j = \quant_{\ell_{j-1}+1} \pi_{\ell_{j-1}+1} \in X_{\ell_{j-1}+1} \cdots \quant_{\ell_{j}} \pi_{\ell_{j}} \in X_{\ell_{j}} 
    $
    of trace quantifiers (with $\ell_0 = 0$). As $\phi$ is a sentence, this implies that we have $\set{X_{\ell_j+1}, \ldots, X_{\ell_{j}}} \subseteq \set{\univar,\unidisvar, Y_1, \ldots, Y_{j-1}}$.

    \item The free variables of $\psi_j^\con$ are among the trace variables quantified in the $\gamma_{j'}$ and $\univar,\unidisvar, Y_1, \ldots, Y_j$.
    
    \item $\psi$ is a quantifier-free formula. Again, as $\phi$ is a sentence, the free variables of $\psi$ are among the trace variables quantified in the $\gamma_j$.
\end{itemize}

Now, $\phi$ is an $\lfpsohyltlfp$ sentence\footnote{Our definition here differs slightly from the one of  \cite{DBLP:conf/cav/BeutnerFFM23} in that we allow to express the existence of some traces in the fixed point (via the subformulas~$\overdot{\pi}_i\tracein Y_j$). All examples and applications of~\cite{DBLP:conf/cav/BeutnerFFM23} are also of this form.}, if additionally each $\phi_j^\con$ has the form 
    \[ \phi_j^\con = \overdot{\pi}_1\tracein Y_j \wedge \cdots \wedge \overdot{\pi}_n \tracein Y_j \wedge   \forall \overdotdot{\pi}_1 \in Z_1.\ \ldots \forall \overdotdot{\pi}_{n'} \in Z_{n'}.\ \psi^\step_j \rightarrow \overdotdot{\pi}_m \tracein Y_j \]
 for some $n\ge 0$, $n' \ge 1$, where $1 \le m \le {n'}$, and where we have
\begin{itemize}
    \item $\set{\overdot{\pi}_1,\ldots, \overdot{\pi}_n} \subseteq \set{\pi_1, \ldots, \pi_{\ell_j} }$,
    \item $\set{Z_1, \ldots, Z_{n'}} \subseteq \set{\univar, \unidisvar, Y_1, \ldots, Y_j}$, and
    \item $\psi^\step_j$ is quantifier-free with free variables among $\overdotdot{\pi}_1, \ldots, \overdotdot{\pi}_{n'}, \pi_1, \ldots, \pi_{\ell_j}$.
\end{itemize} 
As always, $\phi_j^\con$ can be brought into the required prenex normal form.

Let us give some intuition for the definition. To this end, fix some~$j \in \set{1,2,\ldots, k}$ and a variable assignment~$\Pi$ whose domain contains at least all variables quantified before $Y_j$, i.e., all $Y_{j'}$ and all variables in the $\gamma_{j'}$ for $j' < j$, as well as $\univar$ and $\unidisvar$. 
Then, 
\[ \phi_j^\con= \overdot{\pi}_1\in Y_j \wedge \cdots \wedge \overdot{\pi}_n \in Y_j \wedge \left(  \forall \overdotdot{\pi}_1 \in Z_1.\ \ldots \forall \overdotdot{\pi}_{n'} \in Z_{n'}.\ \psi^\step_j \rightarrow \overdotdot{\pi}_m \tracein Y_j\right) \]
induces the monotonic function~$f_{\Pi,j} \colon \pow{(\pow{\ap})^\omega} \rightarrow \pow{(\pow{\ap})^\omega}$ defined as
\begin{multline*}
f_{\Pi,j}(S) = S \cup \set{\Pi(\overdot{\pi}_1), \ldots, \Pi(\overdot{\pi}_n)} \cup  \set{\Pi'(\overdotdot{\pi}_m) \mid \Pi' = \Pi[\overdotdot{\pi}_1 \mapsto t_1, \ldots, \overdotdot{\pi}_{n'} \mapsto t_{n'}]\\
 \text{ for } t_i \in \Pi(Z_i) \text{ if } Z_i \neq Y_j \text{ and }t_i \in S \text{ if } Z_i = Y_j \text{ s.t. }\Pi' \models \psi^\step_j}.    
\end{multline*}
We define $S_0 = \emptyset$, $S_{\ell+1} =  f_{\Pi,j}(S_\ell)$, and 
\[\lfp(\Pi,j) = \bigcup\nolimits_{\ell\in\nats} S_\ell,\] 
which is the least fixed point of $f_{\Pi,j}$. Due to the minimality constraint on $Y_j$ in $\phi$, $\lfp(\Pi,j)$ is the unique set in $\solutions(\Pi, (Y_j,\smallest,\phi^\con_j))$. Hence, an induction shows that $\lfp(\Pi,j)$ only depends on the values~$\Pi(\pi)$ for trace variables~$\pi$ quantified before $Y_j$ as well as the values~$\Pi(\unidisvar)$ and $\Pi(\univar)$, but not on the values~$\Pi(Y_{j'})$ for $j' < j$ (as they are unique).

Thus, as $\solutions(\Pi, (Y_j,\smallest,\phi^\con_j))$ is a singleton, it is irrelevant whether $\quant_j$ is an existential or a universal quantifier. Instead of interpreting second-order quantification as existential or universal, here one should understand it as a deterministic least fixed point computation: choices for the trace variables and the two distinguished second-order variables uniquely determine the set of traces that a second-order quantifier assigns to a second-order variable.

\begin{remark}
\label{remark_lfpcw}
Note that the traces that are added to a fixed point assigned to $Y_j$ either come from another $Y_{j'}$ with $j' < j$, from the model (via $\unidisvar$), or from the set of all traces (via $\univar$). 
Thus, for $\univar$-free formulas, all second-order quantifiers range over (unique) subsets of the model, i.e., there is no need for an explicit definition of closed-world semantics. The analogue of closed-world semantics for $\lfpsohyltlfp$ is to restrict oneself to $\univar$-free sentences.
\end{remark}

\subsection{Satisfiability under Standard Semantics}

Recall that $\lfpsohyltlfp$ satisfiability for $\univar$-free sentences is $\Sigma_1^1$-hard. 
Here, the lower bound already holds for the fragment $\hyltl$~\cite{hyperltlsatconf}, so the more interesting result is the upper bound.
It is obtained by showing that each such sentence has a countable model, and then showing that the existence of a countable model can be captured in $\Sigma_1^1$.

The upper bound on the size of models is obtained by showing that every model of a satisfiable sentence contains a countable subset that is also a model of the sentence.
Intuitively, one takes a minimal set that is closed under application of Skolem functions and shows that it has the desired properties.
This is correct, as in $\univar$-free sentences, only traces from the model are quantified over.
On the other hand, the sentence~$\forall \pi \in \univar.\ \exists \pi'\in\unidisvar.\ \equals{\pi}{\pi'}{\ap}$ has only uncountable models, if $\size{\ap}>1$.

Thus, in $\lfpsohyltlfp$ formulas with $\univar$, one can refer to all traces and thus mimic quantification over sets of natural numbers. 
Furthermore, the satisfiability problem asks for the existence of a model. This \emph{implicit} existential quantifier can be used to mimic existential quantification over sets of sets of natural numbers.
Together with the fact that one can implement addition and multiplication in $\hyltl$, one can show that $\lfpsohyltlfp$ satisfiability for sentences with $\univar$ is $\Sigma_1^2$-hard.

To prove a matching upper bound, we capture the existence of a model and Skolem functions witnessing that it is indeed a model in $\Sigma_1^2$.
Here, the challenge is to capture the least fixed points when mimicking the second-order quantification of $\lfpsohyltlfp$. Naively, this requires an existential quantifier (\myquot{there exists a set that satisfies the guard}) followed by a universal one (\myquot{all strict subsets do not satisfy the guard}).
However, as traces are encoded as sets, this would require universal quantification of type~$3$ objects.
Thus, this approach is not sufficient to prove a $\Sigma_1^2$ upper bound.
Instead, we do not explicitly quantify the fixed points, but instead use witnesses for the membership of traces in the fixed points.
This is sufficient, as the sets of traces quantified in $\lfpsohyltlfp$ are only used as ranges for trace quantification.

\begin{theorem}
\label{thm_satcomplexity_lfp_ss}
$\lfpsohyltlfp$ satisfiability is $\Sigma^2_1$-complete.
\end{theorem}

\begin{proof}
We begin with the lower bound. Let $S \in \Sigma_1^2$, i.e., there exists a formula of arithmetic of the form
\[
\phi(x) = \exists \mathcal{Y}_1\subseteq \pow{\nats}.\ \cdots \exists \mathcal{Y}_k\subseteq \pow{\nats}. \ \psi(x, \mathcal{Y}_1, \ldots,\mathcal{Y}_k)
\]
with a single free (first-order) variable~$x$ such that $S = \set{n \in\nats \mid \natsstruct\models\phi(n) }$, where $\psi$ is a formula with arbitrary quantification over type~$0$ and type~$1$ objects (but no third-order quantifiers) and free third-order variables~$\mathcal{Y}_i$, in addition to the free first-order variable~$x$.
We present a polynomial-time translation from natural numbers~$n$ to $\lfpsohyltlfp$ sentences~$\phi_n$ such that $n \in S$ (i.e., $\natsstruct \models \phi(n)$) if and only if $\phi_n$ is satisfiable.
This implies that $\lfpsohyltlfp$ satisfiability is $\Sigma_1^2$-hard.

Intuitively, we ensure that each model of $\phi_n$ contains enough traces to encode each subset of $\nats$ by a trace (this requires the use of $\univar$). 
Furthermore, we have additional propositions~$\marker_i$, one for each third-order variable~$\mathcal{Y}_i$ existentially quantified in $\phi$, to label the traces encodings sets.
Thus, the set of traces marked by $\marker_i$ encodes a set of sets, i.e., we have mimicked existential third-order quantification by quantifying over potential models. 
Hence, it remains to mimic quantification over natural numbers and sets of natural numbers as well as addition and multiplication, which can all be done in $\hyltl$: quantification over traces mimics quantification over sets and singleton sets (i.e., numbers) and addition and multiplication can be implemented in $\hyltl$.

Let $\ap = \set{\inprop} \cup \ap_\marker \cup \ap_\arith$ with  $\ap_\marker = \set{\marker_1,\ldots, \marker_k}$ and $\ap_\arith = \set{ \argone, \argtwo, \res, \add, \mult}$ and consider the following two formulas (both with a free variable~$\pi'$):
\begin{itemize}
    \item $\psi_0 = \left(\bigwedge_{\prop \in \ap_\arith} \G\neg\prop_{\pi'}\right) \wedge \left(\bigwedge_{i=1}^k (\G (\marker_i)_{\pi'}) \vee (\G\neg(\marker_i)_{\pi'})\right)$ expressing that the interpretation of $\pi'$ may not contain any propositions from $\ap_\arith$ and, for each $i \in [k]$, is either marked by $\marker_i$ (if $\marker_i$ holds at every position) or is not marked by $\marker_i$ (if $\marker_i$ holds at no position). Note that there is no restriction on the proposition~$\inprop$, which therefore encodes a set of natural numbers on each trace. 
    Thus, we can use trace quantification to mimic quantification over sets of natural numbers and quantification of natural numbers (via singleton sets).
    In our encoding, a trace bound to some variable~$\pi$ encodes a singleton set if and only if the formula~$(\neg \inprop_{\pi})\U(\inprop_{\pi} \wedge \X\G\neg \inprop_{\pi}) $ is satisfied.
    
    As explained above, we use the markings to encode the membership of such sets in the $\mathcal{Y}_i$, thereby mimicking the existential quantification of the $\mathcal{Y}_i$.

    However, we need to ensure that this marking is done consistently, i.e., there is no trace~$t$ over $\set{\inprop}$ in the model that is both marked by $m_i$ and not marked by $m_i$. 
    This could happen, as these are just two different traces over $\ap$. 
    However, the formula
    \[
    \psi_\consistent = \forall \pi \in \unidisvar.\ \forall \pi' \in \unidisvar.\ \equals{\pi}{\pi'}{\set{x}} \rightarrow \equals{\pi}{\pi'}{\set{\inprop} \cup \ap_\marker}
    \]
    disallows this.
    
    \item $\psi_1 = \bigwedge_{\prop \in \set{\inprop} \cup \ap_\marker}\G\neg\prop_{\pi'}$ expresses that the interpretation of $\pi'$ may only contain propositions in $\ap_\arith$. We use such traces to implement addition and multiplication in $\hyltl$.
\end{itemize}

To this end, let $T_\plustimes$ be the set of all traces $t \in (\pow{\ap})^\omega$ such that
\begin{itemize}

    \item there are unique $n_1, n_2, n_3 \in \nats$  with $\argone \in t(n_1)$, $\argtwo \in t(n_2)$, and $\res \in t(n_3)$, and

    \item either 
    \begin{itemize}
        \item $\add \in t(n)$ and $\mult \notin t(n)$ for all $n$, and $n_1+n_2 = n_3$, or
        \item $\mult \in t(n)$ and $\add \notin t(n)$ for all $n$, and $n_1 \cdot n_2 = n_3$.
    \end{itemize}

\end{itemize}
there exists a satisfiable $\hyltl$ sentence $\phi_\plustimes$ such that the $\set{\argone, \argtwo, \res, \add, \mult}$-projection of every model of $\phi_\plustimes$ is $T_\plustimes$~\cite[Theorem 5.5]{hyperltlsat}.
As $\hyltl$ is a fragment of $\lfpsohyltlfp$, we can use $\phi_\plustimes$ to construct our desired formula. Furthermore, by closely inspecting the formula~$\phi_\plustimes$, we can see that it can be brought into the form required for guards in $\lfpsohyltlfp$. Call the resulting formula~$\phi_\plustimes'$. It uses two free variables~$\pi_\add$ and $\pi_\mult$ as \myquot{seeds} for the fixed point computation and comes with another $\ltl$ formula~$\psi_s$ with free variables~$\pi_\add$ and $\pi_\mult$ that ensures that the seeds have the right format.

Now, given $n\in\nats$ we define the $\lfpsohyltlfp$ sentence
\begin{align*}
\phi_n = {}&{} \psi_\consistent \wedge \forall\pi \in\univar.\ \left(
(\exists \pi' \in \unidisvar.\ \equals{\pi}{\pi'}{\set{\inprop}} \wedge \psi_0) \wedge
(\exists \pi' \in \unidisvar.\ \equals{\pi}{\pi'}{\ap_\arith} \wedge \psi_1)\right)
\wedge \\
{}&{}\quad\exists \pi_\add \in \unidisvar.\ \exists \pi_\mult \in \unidisvar.\ \psi_s \wedge \exists(X_\arith, \smallest, \phi_\arith').\ \hyperize(\exists x.\ ( x = n \wedge \psi)).    
\end{align*}
Intuitively, $\phi_n$ requires that the model contains, for each subset~$S$ of $\nats$, a unique trace encoding $S$ (additionally marked with the $\marker_i$ to encode membership of $S$ in the set bound to $\mathcal{Y}_i$), contains each trace over $\ap_\arith$, the set~$X_\arith$ is interpreted with $T_\plustimes$, and the formula~$\hyperize(\exists x.\ ( x = n \wedge \psi))$ is satisfied, where the translation~$\hyperize$ is defined below.
Note that we use the constant~$n$, which is definable in first-order arithmetic (with a formula that is polynomial in $\log(n)$ using the fact that the constant~$2$ is definable in first-order arithmetic and then using powers of $2$ to define~$n$) and where $\hyperize$ is defined inductively as follows: 
\begin{itemize}
    \item For second-order variables~$Y$, $\hyperize(\exists Y.\ \psi) = \exists \pi_Y \in \unidisvar.\ \neg\add_{\pi_Y} \wedge \neg\mult_{\pi_Y} \wedge \hyperize(\psi)$, as only traces not being labeled by $\add$ or $\mult$ encode sets.

    \item For second-order variables~$Y$, $\hyperize(\forall Y.\ \psi) = \forall \pi_Y \in \unidisvar.\ (\neg\add_{\pi_Y} \wedge \neg\mult_{\pi_Y}) \rightarrow \hyperize(\psi)$.

    \item For first-order variables~$y$, $\hyperize(\exists y.\ \psi) = \exists \pi_y \in \unidisvar.\ \neg\add_{\pi_y} \wedge \neg\mult_{\pi_y} \wedge ((\neg \inprop_{\pi_y})\U(\inprop_{\pi_y} \wedge \X\G\neg \inprop_{\pi_y})) \wedge \hyperize(\psi)$.

    \item For first-order variables~$y$, $\hyperize(\forall y.\ \psi) = \forall \pi_y \in \unidisvar.\ (\neg\add_{\pi_y} \wedge \neg\mult_{\pi_y} \wedge (\neg \inprop_{\pi_y})\U(\inprop_{\pi_y} \wedge \X\G\neg \inprop_{\pi_y})) \rightarrow \hyperize(\psi)$.

    \item $\hyperize(\psi_1 \lor \phi_2) = \hyperize(\psi_1) \lor \hyperize(\psi_2)$.
    
    \item $\hyperize(\lnot \psi) = \lnot \hyperize(\psi) $.

    \item For third-order variables~$\mathcal{Y}_i$ and second-order variables~$Y$, $\hyperize(Y\in \mathcal{Y}_i) = (\marker_i)_{\pi_Y}$.
    
    \item For second-order variables~$Y$ and first-order variables~$y$, $\hyperize(y\in Y) = \F(\inprop_y \land \inprop_{\pi_Y})$.
    
    \item For first-order variables~$y,y'$, $\hyperize(y<y') = \F(\inprop_y \land \X\F\inprop_{y'})$.
    
    \item For first-order variables~$y_1,y_2,y$, $\hyperize(y_1+y_2=y) = \exists \pi \in X_\arith.\ \add_\pi \land \F(\inprop_{\pi_{y_1}}\land\argone_\pi) \land \F(\inprop_{\pi_{y_2}}\land\argtwo_\pi) \land \F(\inprop_{\pi_y}\land\res_\pi)$.
    
    \item For first-order variables~$y_1,y_2,y$, $\hyperize(y_1 \cdot y_2=y) = \exists \pi \in X_\arith.\ \mult_\pi \land \F(\inprop_{\pi_{y_1}}\land\argone_\pi) \land \F(\inprop_{\pi_{y_2}}\land\argtwo_\pi) \land \F(\inprop_{\pi_y}\land\res_\pi)$.
    
\end{itemize}
While $\phi_n$ is not in prenex normal form, it can easily be brought into prenex normal form, as there are no quantifiers under the scope of a temporal operator.
An induction shows that we indeed have that $\natsstruct\models\phi(n)$ if and only if $\phi_n$ is satisfiable.
    
For the upper bound, we show that $\lfpsohyltlfp$ satisfiability is in $\Sigma_1^2$. More formally, we show how to construct a formula of the form
\[
\theta(x) = \exists \mathcal{Y}_1\subseteq \pow{\nats}.\ \cdots \exists \mathcal{Y}_k\subseteq \pow{\nats}. \ \psi(x, \mathcal{Y}_1, \ldots,\mathcal{Y}_k)
\]
with a single free (first-order) variable~$x$ such that an $\lfpsohyltlfp$ sentence~$\phi$ is satisfiable if and only if $\natsstruct \models \theta(\encode{\phi})$.
Here, $\psi$ is a formula of arithmetic with arbitrary quantification over type~$0$ and type~$1$ objects (but no third-order quantifiers) and free third-order variables~$\mathcal{Y}_i$, in addition to the free first-order variable~$x$, and $\encode{\cdot}$ is a polynomial-time computable injective function mapping $\lfpsohyltlfp$ sentences to natural numbers.

In the following, we assume, without loss of generality, that $\ap$ is fixed, so that we can use $\size{\ap}$ as a constant in our formulas (which is definable in arithmetic). 
Then, we can encode traces as sets of natural numbers. To do to, we need to introduce some notation following Frenkel and Zimmermann~\cite{sohypercomplexity}:
Let $\pair \colon \nats\times\nats \rightarrow\nats$ denote Cantor's pairing function defined as $\pair(i,j) = \frac{1}{2}(i+j)(i+j+1) +j$, which is a bijection and can be implemented in arithmetic.
Furthermore, we fix a bijection~$e \colon \ap \rightarrow\set{0,1,\ldots,\size{\ap}-1}$.
Then, we encode a trace~$t \in (\pow{\ap})^\omega$ by the set~$S_t =\set{\pair(j,e(\prop)) \mid j \in \nats \text{ and } \prop \in t(j)} \subseteq \nats$.
Now, the formula
\[\phi_\istrace(Y) = \forall x.\ \forall y.\ y \ge \size{\ap} \rightarrow \pair(x,y) \notin Y \]
is satisfied in $\natsstruct$ if and only if the interpretation of $Y$ encodes a trace over $\ap$~\cite{sohypercomplexity}.
Furthermore, a finite collection of sets~$S_1, \ldots, S_k$ is uniquely encoded by the set~$\set{\pair(n,j) \mid n \in S_j}$, i.e., we can encode finite sets of sets by type~$1$ objects.
In particular, we can encode a variable assignment whose domain is finite and contains only trace variables by a set of natural numbers and we can write a formula that checks whether a trace (encoded by a set) is assigned to a certain variable.

Now, the overall proof idea is to let the formula~$\theta$ of arithmetic express the existence of Skolem functions for the existentially quantified variables in the $\lfpsohyltlfp$ sentence~$\phi$ such that each variable assignment that is consistent with the Skolem functions satisfies the maximal quantifier-free subformula of $\varphi$.
For trace variables~$\pi$, a Skolem function is a type~$2$ object, i.e., a function mapping a tuple of sets of natural numbers (encoding a tuple of traces, one for each variable quantified universally before $\pi$) to a set of natural numbers (encoding a trace for $\pi$).
However, to express that the interpretation of a second-order variable~$X$ is indeed a least fixed point we need both existential quantification (\myquot{there exists a set} that satisfies the guard) and universal quantification (\myquot{every other set that satisfies the guard is larger}).
Thus, handling second-order quantification this way does not yield the desired~$\Sigma_1^2$ upper bound. 

Instead we use that fact that membership of a trace in an $\hyltl$-definable least fixed point can be witnessed by a finite tree labeled by traces, i.e., by a type~$1$ object.
Thus, instead of capturing the full least fixed point in arithmetic, we verify on-the-fly for each trace quantification of the form~$\exists \pi \in Y_j$ or $\forall \pi \in Y_j$ whether the interpretation of $\pi$ is in the interpretation of $Y_j$, which only requires the quantification of a witness tree.

Before we can introduce these witness trees, we need to introduce some additional notation to express satisfaction of quantifier-free $\lfpsohyltlfp$ formulas in arithmetic.
To this end, fix some \[
\phi = 
\gamma_1 \quant_1(Y_1, \smallest, \phi_1^\con).\ 
\gamma_2 \quant_2(Y_2, \smallest, \phi_2^\con).\ 
\ldots 
\gamma_k \quant_2(Y_k, \smallest, \phi_k^\con).\
\gamma_{k+1}.\
\psi,
\]
where $\psi$ is quantifier-free.
We assume w.l.o.g.\ that each trace variable is quantified at most once in $\varphi$, which can always be achieved by renaming variables. This implies that for each trace variable~$\pi$ quantified in some $\gamma_j$, there exists a unique second-order variable~$X_\pi$ such that $\pi$ ranges over $X_\pi$. 
Furthermore, we assume that each $Y_j$ is different from $\unidisvar$ and $\univar$, which can again be achieved by renaming variables, if necessary.
Recall that the values of the least fixed points are uniquely determined by the interpretations of $\unidisvar$, $\univar$, and the trace variables in the $\gamma_j$. 
We say that a variable assignment $\Pi$ is $\phi$-sufficient if $\Pi$'s domain contains exactly these variables.

Let $\xi$ be a quantifier-free formula and let $\Pi$ be a variable assignment whose domain contains at least all free variables of $\xi$, and let $\Xi$ be the set of subformulas of $\xi$.
The expansion of $\xi$ with respect to $\Pi$ is the function~$e_{\Pi,\xi}\colon \Xi\times\nats \rightarrow \set{0,1}$ defined as $e_{\Pi,\xi}(\xi',j) = 1$ if and only if $\suffix{\Pi}{j} \models\xi'$. 
From the semantics of $\lfpsohyltlfp$, it follows that $e_{\Pi,\xi}$ is the unique function~$e\colon \Xi\times\nats \rightarrow \set{0,1}$ satisfying the following properties:
\begin{itemize}
    \item $e(\prop_\pi,j) = 1$  if and only if $\prop \in \Pi(\pi)(j)$.
    \item $e(\neg \xi',j) = 1$ if and only if $e(\xi',j) = 0$.
    \item $e(\xi_1 \vee \xi_2,j) = 1$ if and only if $e(\xi_1,j) =1$ or $e(\xi_2,j) =1$.
    \item $e(\X\xi',j) = e(\xi',j+1)$.
    \item $e(\xi_1\U\xi_2,j) = 1$ if and only if there exists a $j' \ge j$ such that $e(\xi_2,j') = 1$ and $e(\xi_1,j'') = 1$ for all $j \le j'' < j'$.
\end{itemize}
Let us remark here for further use later that a function from $\Xi\times\nats $ to $\set{0,1}$ is a type~$1$ object (as functions from $\nats$ to $\nats$ can be encoded as subsets of $\nats$ via their graph and the pairing function introduced below) and that all five requirements above can then be expressed in first-order logic.

Now, we introduce witness trees following Frenkel and Zimmermann~\cite{arxiv}.
Consider the formula
\begin{equation}
\label{eq_phicon}
\phi_j^\con = \overdot{\pi}_1\tracein Y_j \wedge \cdots \wedge \overdot{\pi}_n \tracein Y_j \wedge   \forall \overdotdot{\pi}_1 \in Z_1.\ \ldots \forall \overdotdot{\pi}_{n'} \in Z_{n'}.\ \psi^\step_j \rightarrow \overdotdot{\pi}_m \tracein Y_j.    
\end{equation}
inducing the unique least fixed point that $Y_j$ ranges over. 
It expresses that a trace~$t$ is in the fixed point either because it is of the form~$\Pi(\overdot{\pi}_i)$ for some $i \in \set{1,\ldots, n}$ where $\overdot{\pi}_i$ is a trace variable quantified before the quantification of $Y_j$, or $t$ is in the fixed point because there are traces~$t_1,\ldots, t_{n'}$ such that assigning them to the $\overdotdot{\pi}_i$ satisfies $\psi^\step_j$ and $t = t_m$. 
Thus, the traces~$t_1,\ldots, t_{n'}$ witness that $t$ is in the fixed point. 
However, each $t_i$ must be selected from $\Pi(Z_i)$, which, if $Z_i = Y_{j'}$ for some $j'$, again needs a witness.
Thus, a witness is in general a tree labeled by traces and indices from $\set{1, \ldots, k}$ denoting the fixed point variable the tree proves membership in.
Note that $\psi^\step_j$ may contain free variables that are quantified in the $\gamma_{j'}$ for $1 \le j' \le j$. 
The membership of such variables~$\pi$ in $\Pi(X_\pi)$ will not be witnessed by this witness tree, as their values are already determined by $\Pi$.

Formally, let us fix a $j^* \in \set{1,\ldots,k}$, a $\phi$-sufficient assignment~$\Pi$, and a trace~$t^*$. 
A $\Pi$-witness tree for $(t^*,j^*)$ (which witnesses~$t^* \in\lfp(\Pi,j^*)$) is an ordered finite tree whose vertices are labeled by pairs~$(t,j)$ where $t$ is a trace and where $j $ is in $\set{1,\ldots, k} \cup \set{a,d}$ (where $a$ and $d$ are fresh symbols) such that the following are satisfied:
\begin{itemize}
    \item The root is labeled by $(t^*, j^*)$.
    
    \item If a vertex is labeled with $(t,a)$ for some trace~$t$, then it must be a leaf. Note that $t$ is in $\Pi(\univar)$, as that set contains all traces.
    
    \item If a vertex is labeled with $(t,d)$ for some trace~$t$, then it must be a leaf and $t \in \Pi(\unidisvar)$.
    
    \item Let $(t,j)$ be the label of some vertex~$v$ with $j \in \set{1,\ldots, k}$ and let $\phi_j^\con$ as in (\ref{eq_phicon}).
If $v$ is a leaf, then we must have $t^* = \Pi(\overdot{\pi}_i)$ for some $i \in \set{1, \ldots, n}$.
If $v$ is an internal vertex, then it must have $n'$ successors labeled by $(t_1,j_1),\ldots,(t_{n'},j_{n'})$ (in that order) such that $\Pi[\overdotdot{\pi}_1 \mapsto t_1, \ldots, \overdotdot{\pi}_{n'} \mapsto t_{n'}] \models \psi^\step_j $, $t = t_m$, and
such that the following holds for all $i \in \set{1, \ldots, n'}$:
\begin{itemize}
    \item If $Z_i = \univar$, then the $i$-th successor of $v$ is a leaf and $j_i = a$.    
    
    \item If $Z_i = \unidisvar$, then the $i$-th successor of $v$ is a leaf, $t_i \in \Pi(\unidisvar)$, and $j_i = d$.
    
    \item If $Z_i = Y_{j'}$ for some $j'$ (which satisfies $1 \le j' \le j$ by definition of $\lfpsohyltlfp$), then we must have that $ j_i = j'$ and that the subtree rooted at the $i$-th successor of $v$ is a $\Pi$-witness tree for $(t_i,j_i)$.
\end{itemize}
\end{itemize}

The following proposition states that membership in the fixed points is witnessed by witness trees. It is obtained by generalizing a similar argument for $\univar$-free sentences to sentences (potentially) using $\univar$.

\begin{proposition}[cf.~\cite{arxiv}]
Let $\Pi$ be a variable assignment whose domain contains all variables in the $\gamma_j$ and satisfies $\Pi(Y_j) = \lfp(\Pi, j)$ for all $j$. Then, we have $t \in \Pi(Y_j)$ if and only if there exists a $\Pi$-witness tree for $(t,j)$.    
\end{proposition}

Note that a witness tree is a type~$1$ object: The (finite) tree structure can be encoded by 
\begin{itemize}
    \item a natural number~$s>0$ (encoding the number of vertices), 
    \item a function from $\set{1,\ldots, s} \times \set{1,\ldots, n'} \rightarrow \set{0,1,\ldots, s}$ encoding the child relation, i.e., $(v,j) \mapsto v'$ if and only if the $j$-th child of $v$ is $v'$ (where we use $0$ for undefined children), 
    \item $s$ traces over $\ap$ and $s$ values in $\set{1,\ldots, k, k+1, k+2}$ to encode the labeling (where we use $k+1$ for $a$ and $k+2$ for $d$).
\end{itemize}
Note that the function encoding the child relation can be encoded by a finite set by encoding its graph using the pairing function while all other objects can directly be encoded by sets of natural numbers, and thus be encoded by a single set as explained above.

Furthermore, one can write a second-order formula~$\psi_{\hastree}(X_D,A, X_{t^*}, j^*)$ with free third-order variable~$X_D$ (encoding a set of traces~$T$), free second-order variables~$A$ (encoding a variable assignment~$\Pi$ whose domain contains exactly the trace variables in the $\gamma_j$) and $X_{t^*}$ (encoding a trace~$t^*$), and free first-order variable~$j^*$ that holds in $\natsstruct$ if and only if there exists a $\Pi[\unidisvar\mapsto T, \univar\mapsto (\pow{\ap})^\omega]$-witness tree for $(t^*,j^*)$.
To evaluate the formulas~$\psi_{j'}^\step$ as required by the definition of witness trees, we rely on the expansion as introduced above, which here is a mapping from vertices in the tree and subformulas of the $\psi_{j'}^\step$ to $\set{0,1}$, and depends on the set of traces encoded by $X_D$ and the variable assignment encoded by $A$.
Such a function is a type~$1$ object and can therefore be quantified in $\psi_{\hastree}(X_D,A, X_{t^*}, j^*)$.

Recall that we construct a formula~$\theta(x)$ with a free first-order variable~$x$ such that an $\lfpsohyltlfp$ sentence~$\phi$ is satisfiable if and only if $\natsstruct \models \theta(\encode{\phi})$, where $\encode{\cdot}$ is a polynomial-time computable injective function mapping $\lfpsohyltlfp$ sentences to natural numbers.

With this preparation, we can define $\theta(x)$ such that it is not satisfied in $\natsstruct$ if the interpretation of $x$ does not encode an $\lfpsohyltlfp$ sentence.
If it does encode such a sentence~$\phi$, let $\psi$ be its maximal quantifier-free subformula (which are \myquot{computable} in first-order arithmetic using a suitable encoding~$\encode{\cdot}$).
Then, $\theta$ expresses the existence of
\begin{itemize}
    \item a set~$T$ of traces (bound to the third-order variable~$X_D$ and encoded as a set of sets of natural numbers, i.e., a type~$2$ object),
    \item Skolem functions for the existentially quantified trace variables in the $\gamma_j$ (which can be encoded by functions from $(\pow{\nats})^\ell$ to $\pow{\nats}$ for some $\ell$, i.e., by a type~$2$ object), and
    \item a function~$e \colon \pow{\nats} \times \nats \times \nats \rightarrow \set{0,1}$
\end{itemize}
such that the following is true for all variable assignments~$\Pi$ (restricted to the trace variables in the $\gamma_j$ and bound to the second-order variable~$A$):
If 
\begin{itemize}
    \item $\Pi$ is consistent with the Skolem functions for all existentially quantified variables,
    \item for all universally quantified $\pi$ in some $\gamma_j$ with $X_\pi = \unidisvar$ we have $\Pi(\pi) \in T$, and
    \item for all universally quantified $\pi$ in some $\gamma_j$ with $X_\pi = Y_{j^*}$ the formula~$\psi_{\hastree}(X_D,A, \Pi(\pi), j^*)$ holds, 
\end{itemize}
then
\begin{itemize}
    \item for all existentially quantified $\pi$ in some $\gamma_j$ with $X_\pi = \unidisvar$ we have $\Pi(\pi) \in T$, and
    \item for all existentially quantified $\pi$ in some $\gamma_j$ with $X_\pi = Y_{j^*}$ the formula~$\psi_{\hastree}(X_D,A, \Pi(\pi), j^*)$ holds, and
    \item the function~$(x,y)\mapsto e(A,x,y)$ is the expansion of $\psi$ with respect to $\Pi$ and we have $e(A,\psi,0) =1$ (here we identify subformulas of $\psi$ by natural numbers).
\end{itemize}
We leave the tedious, but routine, details to the reader. 
\end{proof}

Note that in the lower bound proof, a single second-order quantifier (for the set of traces implementing addition and multiplication) suffices. 

\subsection{Finite-State Satisfiability and Model-Checking}

Recall that $\lfpsohyltlfp$ finite-state satisfiability and model-checking are known to be $\Sigma_1^1$-hard and in $\Sigma_2^2$, where the lower bounds hold for $\univar$-free sentences and the upper bounds for sentences with $\univar$.
Here, we close these gaps by showing that both problems are equivalent to truth in second-order arithmetic.

\begin{theorem}
\label{thm_fssatmccomplexity_lfp}
$\lfpsohyltlfp$ finite-state satisfiability and model-checking are polynomial-time equivalent to truth in second-order arithmetic. The lower bounds already hold for $\univar$-free formulas.
\end{theorem}

The result follows from the following building blocks, which are visualized in Figure~\ref{fig_reductions}:

\begin{itemize}
    \item All lower bounds proven for $\univar$-free $\lfpsohyltlfp$ sentences also hold for $\lfpsohyltlfp$ while upper bounds for $\lfpsohyltlfp$ also hold for the fragment of $\univar$-free $\lfpsohyltlfp$ sentences.
    \item $\lfpsohyltlfp$ finite-state satisfiability can in polynomial time be reduced to $\lfpsohyltlfp$ model-checking for $\univar$-free sentences (see Lemma~\ref{lemma_fssat2mc} below).
    \item $\lfpsohyltlfp$ finite-state satisfiability for $\univar$-free sentences is at least as hard as truth in second-order arithmetic (see Lemma~\ref{lemma_fssat_soahard} below).
    \item $\lfpsohyltlfp$ model-checking can in polynomial time be reduced to truth in second-order arithmetic (see Lemma~\ref{lemma_mc_soaeasy} below).
\end{itemize}

\begin{figure}[h]
    \centering
    \begin{tikzpicture}

\fill[lightgray!40,rounded corners] (-5.25,-.4) rectangle (5.25,.4);    
\fill[lightgray!40,rounded corners] (-5.25,2.6) rectangle (5.25,3.4);    

\fill[lightgray!40,rounded corners] (-.75,-1) rectangle (.75,5.75);    
\fill[lightgray!40,rounded corners] (3.25,-1) rectangle (4.75,5.75);    

\fill[lightgray] (-.75,-.4) rectangle (.75,.4);    
\fill[lightgray] (-.75,2.6) rectangle (.75,3.4);    

\fill[lightgray] (3.25,-.4) rectangle (4.75,.4);    
\fill[lightgray] (3.25,2.6) rectangle (4.75,3.4);

\node[anchor = west] at (-5.1,3) {$\lfpsohyltlfp$};
\node[anchor = west] at (-5.1,0) {$\univar$-free $\lfpsohyltlfp$};

\node[align=center,anchor = north] at (0,5.5) {finite-\\state\\satisfi-\\ability};
\node[align=center,anchor = north] at (4,5.5) {model-\\checking};

\fill (0,0) circle (.1cm);
\fill (0,3) circle (.1cm);
\fill (4,0) circle (.1cm);
\fill (4,3) circle (.1cm);

\path[->,>=stealth,ultra thick, shorten > = .2cm, shorten < = .2cm]
(0,0) edge[bend left=15] node[left] {trivial} (0,3)
(4,0) edge[bend right=15] node[right] {trivial} (4,3)
(0,0) edge[bend right=15] node[below] {Lemma~\ref{lemma_fssat2mc}} (4,0)
(0,3) edge[bend left=15] node[above,rotate=-37] {Lemma~\ref{lemma_fssat2mc}} (4,0)
;

\node[align = center, anchor = west] (ub) at (5.75,3) {upper bound:\\Lemma~\ref{lemma_mc_soaeasy}};
\node[align = center, anchor = north] (lb) at (0,-1.5) {lower bound:\\Lemma~\ref{lemma_fssat_soahard}};

\path[->,>=stealth,ultra thick, shorten > = .2cm, shorten < = .2cm, dotted]
(ub) edge (4,3)
(lb) edge (0,0)
;

    \end{tikzpicture}    
    \caption{The reductions (drawn as solid arrows) and lemmata proving Theorem~\ref{thm_fssatmccomplexity_lfp}.}
    \label{fig_reductions}
\end{figure}

We begin with reducing finite-state satisfiability to model-checking.

\begin{lemma}
\label{lemma_fssat2mc}
$\lfpsohyltlfp$ finite-state satisfiability is polynomial-time reducible to $\lfpsohyltlfp$ model-checking for $\univar$-free sentences.
\end{lemma}

\begin{proof}
Intuitively, we reduce finite-state satisfiability to model-checking by existentially quantifying a finite transition system~$\tsys$ and then model-checking against it. 
This requires us to \emph{work} with the set of traces of $\tsys$.
We do so by constructing, using a least fixed point, the set of prefixes of traces of $\tsys$.
This uniquely determines the traces of $\tsys$ as those traces~$t$ for which all of their prefixes are in the fixed point.

Before we begin, we show that we can restrict ourselves, without loss of generality, to finite-state satisfiability by transition systems with a single initial vertex, which simplifies our construction.
Let $\tsys = (V, E, I, \lambda)$ be a transition system and $\phi$ a $\sohyltl$ sentence.
Consider the transition system~$\tsys_{\X} = (V\cup\set{v_\initmark}, E', \set{v_\initmark}, \lambda')$ with a fresh initial vertex~$v_\initmark$,
\[
E' = E \cup \set{(v_\initmark, v_\initmark)} \cup \set{(v_\initmark, v) \mid v \in I},
\]
and $\lambda'(v_\initmark) = \set{\$}$ (for a fresh proposition~$\$$) and $\lambda'(v) = \lambda(v)$ for all $v \in V$.
Here, we add the self-loop on the fresh initial vertex~$v_\initmark$ to deal with the special case of $I$ being empty, which would make $v_\initmark$ terminal without the self-loop.

Now, we have
\[
\traces(\tsys_{\X}) = \set{\$}^\omega \cup \set{\set{\$}^+\cdot t \mid t \in \traces(\tsys) }.
\]
Furthermore, let $\phi_{\X}$ be the formula obtained from $\phi$ by adding an $\X$ to the maximal quantifier-free subformula of $\phi$ and by inductively replacing
\begin{itemize}
\item $\exists \pi \in X.\ \psi$ by $\exists \pi \in X.\ \X(\neg\$_\pi \wedge \psi)$ and
\item $\forall \pi \in X.\ \psi$ by $\forall \pi \in X.\ \X(\neg\$_\pi \rightarrow \psi)$,
\end{itemize}
where $X$ is an arbitrary second-order variable. Then, $\tsys \models \phi$ if and only if $\tsys_{\X} \models \phi_{\X}$.
Thus, $\phi$ is finite-state satisfiable if and only if there exists a finite transition system with a single initial vertex that satisfies $\phi_{\X}$.

Given an $\lfpsohyltlfp$ sentence~$\phi$ (over $\ap$, which we assume to be fixed), we construct a transition system~$\tsys'$ and an $\univar$-free $\lfpsohyltlfp$ sentence~$\phi'$ (both over some $\ap' \supseteq \ap$) such that $\phi$ is satisfied by a finite transition system with a single initial vertex if and only if $\tsys' \models\phi'$.
To this end, we define $\ap' =  \ap \cup \set{\inprop,\#} \cup \set{\argone, \argtwo, \res, \add, \mult}$. The transition system~$\tsys'$ is constructed such that we have $\traces(\tsys') = (\pow{\ap'})^\omega$, which can be achieved with $2^{\size{\ap'}}$ many vertices (which is constant as $\ap$ is fixed!).

Throughout the construction, we use a second-order variable~$Y_\arith$ which will be interpreted by $T_\plustimes$.
In the following, we rely on traces over $\ap'$ of certain forms:
\begin{itemize}
	\item Consider the formula~$\phi_\vertices = \neg \#_{\pi_\vertices} \wedge (\neg \#_{\pi_\vertices}) \U (\#_{\pi_\vertices} \wedge \X\G \neg \#_{\pi_\vertices})$ expressing that the interpretation of ${\pi_\vertices}$ contains a unique position where $\#$ holds. This position may not be the first one. Hence, the $\ap$-projection of such a trace up to the $\#$ is a nonempty word~$w(0)w(1)\cdots w(n-1)$ for some $n > 0$. 
	It induces $V = \set{0,1,\ldots, n-1}$ and $\lambda \colon V \rightarrow \pow{\ap}$ with $\lambda(v) = w(v) \cap \ap$. Furthermore, we fix $I = \set{0}$.
	
	\item Let $t$ be a trace over $\ap'$. It induces the edge relation~$E = \set{(i,j) \in V\times V \mid \inprop \in t(i\cdot n + j)}$, i.e., we consider the first $n^2$ truth values of $\inprop$ in $t$ as an adjacency matrix. Furthermore, if $t$ is the interpretation of $\pi_\edges$ satisfying $\phi_\edges$ defined below, then every vertex has a successor in $E$.
	\begin{align*}
	\phi_\edges ={}&{} \forall \pi \in Y_\arith.\ [\mult_\pi \wedge ((\neg\X \#_{\pi_\vertices}) \U \argone_{\pi} )\wedge \F(\argtwo_\pi \wedge \X\#_{\pi_\vertices})]\rightarrow \\
	{}&{}\quad\exists \pi' \in Y_\arith.\ 
	\add_{\pi'} \wedge 
	\F(\argone_{\pi'} \wedge \res_{\pi}) \wedge 
	(\neg\X\#_{\pi_\vertices})\U (\argtwo_{\pi'}) \wedge 
	\F(\res_{\pi'} \wedge 
	\inprop_{\pi_\edges}).
	\end{align*}
	
	\item Consider the formula
	\[\phi_\prefix = (\neg\#_\pi) \U (\#_\pi \wedge \X\G\neg \#_\pi) \wedge (\X\neg\#_\pi)\U \inprop_\pi
	,\]
	 which is satisfied if the interpretation~$t$ of $\pi$ has a unique position~$\ell$ at which $\#$ holds and if it has a unique position~$v$ where $\inprop$ holds. Furthermore, $v$ must be strictly smaller than $n$, i.e., it is in $V$.
	In this situation, we interpret the $\ap$-projection of $t(0)\cdots t(\ell-1)$ as a trace prefix over $\ap$ and $v$ as a vertex (intuitively, this will be the vertex where the path prefix inducing the trace prefix ends).
	
	Using this encoding, the formula
	 \[
\phi_\init = \phi_\prefix[\pi/\pi_\init] \wedge \inprop_{\pi_\init} \wedge \X\#_{\pi_\init} \wedge \bigwedge_{\prop \in \ap} \prop_{\pi_\vertices} \leftrightarrow \prop_{\pi_\init}
\]
ensures that $\pi_\init$ is interpreted with a trace that encodes the trace prefix of the unique path of length one starting at the initial vertex. Here, $\phi_\prefix[\pi/\pi_\init]$ denotes the formula obtained from $\phi_\prefix$ by replacing each occurrence of $\pi$ by $\pi_\init$.
\end{itemize}

Now, given $\phi$, consider the sentence
\begin{align*}
\phi'' = \exists \pi_\vertices \in \unidisvar.\ \exists \pi_\edges \in \unidisvar.\ \exists \pi_\init \in \unidisvar.\ \phi_\vertices \wedge \phi_\edges \wedge \phi_\init \wedge \exists (Y_\prefixes,\smallest, \phi_\con).\ \rel(\phi),
\end{align*}
where $\phi_\con$ and $\rel(\phi)$ are introduced below.
Before, let us note that $\pi_\vertices$ and $\pi_\edges$ encode a finite transition system~$\tsys$ with a single initial vertex as described above.

Now, $\phi_\con$ is used to ensure that the interpretation of $Y_\prefixes$  must contain exactly the encodings of prefixes of traces of $\tsys$ and defined as 
\begin{align*}
	\phi_\con = \pi_\init \tracein Y_\prefixes \wedge \forall \pi_0 \in Y_\prefixes.\ \forall \pi_1 \in \unidisvar.\ 
	\psi_\step \rightarrow \pi_1 \tracein Y_\prefixes
\end{align*}
with
\[
\psi_\step = \phi_\prefix[\pi/\pi_1] \wedge 
	(\bigwedge_{\prop \in \ap} \prop_{\pi_0} \leftrightarrow \prop_{\pi_1}) \U(\#_{\pi_0} \wedge \X \#_{\pi_1}) \wedge
	\bigwedge_{\prop \in \ap} \F(\prop_{\pi_\vertices} \wedge \inprop_{\pi_1}) \leftrightarrow \F(\prop_{\pi_1}  \wedge \X\#_{\pi_1}) \wedge 
\phi_{\edge},
\]
where $\phi_\edge$ checks whether there is an edge in the encoded transition system~$\tsys$ between the vertex induced by $\pi_0$ and the vertex induced by $\pi_1$:
\begin{align*}
\phi_\edge = \exists \pi_m \in Y_\arith.\ \exists \pi_a\in Y_\arith.\
{}&{}
\mult_{\pi_m} \wedge
\F(\argone_{\pi_m} \wedge \inprop_{\pi_0}) \wedge
\F(\argtwo_{\pi_m} \wedge \#_{\pi_\vertices}) \wedge\\
{}&{}
\add_a \wedge
\F(\argone_{\pi_a} \wedge \res_{\pi_m} ) \wedge
\F(\argtwo_{\pi_a} \wedge \inprop_{\pi_1}) \wedge
\F(\res_{\pi_a} \wedge \inprop_{\pi_\edges})
\end{align*}
Intuitively, the until subformula of $\psi_\step$ ensures that the finite trace encoded by the interpretation of $\pi_0$ is a prefix of the finite trace encoded by the interpretation of $\pi_1$, which has one additional letter.
The equivalence~$\F(\prop_{\pi_\vertices} \wedge \inprop_{\pi_1}) \leftrightarrow \F(\prop_{\pi_1}  \wedge \X\#_{\pi_1})$ then ensures that the additional letter is the label of the vertex~$v$ induced by the interpretation of $\pi_1$: the left-hand side of the equivalence \myquot{looks up} the truth values of the label of $v$ in $\pi_\vertices$ and the right-hand side the truth values at the additional letter in $\pi_1$, which comes right before the $\#$.

Thus, the least fixed point induced by $\psi_\con$ contains exactly the encodings of the prefixes of traces of $\tsys$ as the prefix of length one is included and $\psi_\step$ extends the prefixes by one more letter.

The formula~$\rel(\phi)$ is defined by inductively replacing
\begin{itemize}
    \item each $\exists \pi \in \univar.\ \psi$ by $\exists \pi \in \unidisvar.\ \psi$ (as $\unidisvar$ contains, by construction, all traces over $\ap$), 
    
    \item each $\forall \pi \in \univar.\ \psi$ by $\forall \pi \in \unidisvar.\ \psi$, 
    
    \item each $\exists \pi \in \unidisvar.\ \psi$ by $\exists \pi \in \unidisvar.\ \psi_r \wedge \psi$, and 
    
    \item each $\forall \pi \in \unidisvar.\ \psi$ by $\forall \pi \in \unidisvar.\ \psi_r \rightarrow \psi$,
\end{itemize}
where $\psi_r$ expresses that the trace assigned to $\pi$ must be one of the transition system~$\tsys$:
\begin{align*}
\forall \pi' \in \unidisvar.\ \neg\#_{\pi'} \wedge ((\neg\#_{\pi'})\U(\#_{\pi'}\wedge \X\G\neg \#_{\pi'})) \rightarrow  \exists \pi_p\in Y_\prefixes.\ \F(\#_{\pi'} \wedge \#_{\pi_p}) \wedge (\bigwedge_{\prop\in\ap} \prop_{\pi} \leftrightarrow \prop_{\pi_p})\U\#_{\pi_p}
\end{align*}
Here, we use the fact that if all prefixes of a trace are in the prefixes of the traces of $\tsys$, then the trace itself is also one of $\tsys$. Thus, we can require for every $n > 0$, 
there is a trace in $Y_\prefixes$ encoding a prefix of length~$n$ of a trace of $\tsys$ such that $\pi$ has that prefix. 

To finish the proof, we need to ensure that $Y_\arith$ does indeed contain the traces implementing addition and multiplication, as described in the proof of Theorem~\ref{thm_satcomplexity_lfp_ss}: $\tsys'$ is a model of 
\[
\phi' = \exists \pi_\add \in \unidisvar.\ \exists \pi_\mult \in \unidisvar.\ \psi_s \wedge \exists(X_\arith, \smallest, \phi_\arith').\ \phi''
\]
if and only if $\phi$ is satisfied by a finite transition system with a single initial vertex.
\end{proof}

Now, we prove the lower bound for $\lfpsohyltlfp$ finite-state satisfiability.

\begin{lemma}
\label{lemma_fssat_soahard}
Truth in second-order arithmetic is polynomial-time reducible to $\lfpsohyltlfp$ finite-state satisfiability for $\univar$-free sentences.
\end{lemma}

\begin{proof}
Let $\ap = \set{\inprop,\#}$.
We begin by presenting a (satisfiable) $\lfpsohyltlfp$ sentence~$\phi_\prefs$ that ensures that the set of prefixes of the $\set{\inprop}$-projection of each of its models is equal to $(\pow{\set{\inprop}})^*$.
To this end, consider the conjunction~$\phi_\prefs$ of the following formulas:
\begin{itemize}
    \item $\forall \pi \in \unidisvar.\ (\neg\#_\pi\U(\#_\pi \wedge \X\G\neg \#_\pi))$, which expresses that each trace in the model has a unique position where the proposition~$\#$ holds.
    \item $\exists \pi \in \unidisvar.\  \#_{\pi}$, which expresses that each model contains a trace where $\#$ holds at the first position.
    \item $\forall \pi \in \unidisvar.\ \exists \pi' \in \unidisvar.\ (\inprop_\pi \leftrightarrow \inprop_{\pi'}) \U (\#_{\pi} \wedge \neg\inprop_{\pi'} \wedge \X\#_{\pi'} )$ expressing that for every trace~$t$ in the model there is another trace~$t'$ in the model such that the $\set{\inprop}$-projection~$w$ of $t$ up to the $\#$ and the $\set{\inprop}$-projection~$w'$ of $t'$ up to the $\#$ satisfy $w' = w\emptyset$.
    \item $\forall \pi \in \unidisvar.\ \exists \pi' \in \unidisvar.\ (\inprop_\pi \leftrightarrow \inprop_{\pi'}) \U (\#_{\pi} \wedge \inprop_{\pi'} \wedge \X\#_{\pi'} )$ expressing that for every trace~$t$ in the model there is another trace~$t'$ in the model such that the $\set{\inprop}$-projection~$w$ of $t$ up to the $\#$ and the $\set{\inprop}$-projection~$w'$ of $t'$ up to the $\#$ satisfy $w' = w\set{\inprop}$.   
\end{itemize}

A straightforward induction shows that the set of prefixes of the $\set{\inprop}$-projection of each model of $\phi_\prefs$ is equal to $(\pow{\set{\inprop}})^*$.
Furthermore, let $\tsys$ be a finite transition system that is a model of $\phi_\prefs$, i.e., with $\traces(\tsys) \models \phi_\prefs$. 
Then, the $\set{\inprop}$-projection of $\traces(\tsys)$ must be equal to $(\pow{\set{\inprop}})^\omega$, which follows from the fact that the set of traces of a transition system is closed (see, e.g.,~\cite{BK08} for the necessary definitions).
Thus, as usual, we can mimic set quantification over $\nats$ by trace quantification. 

Furthermore, recall that, as in the proof of Theorem~\ref{thm_satcomplexity_lfp_ss}, we can implement addition and multiplication in $\lfpsohyltlfp$: there is an $\lfpsohyltlfp$ guard for a second-order quantifier such that the $\set{\argone, \argtwo, \res, \add, \mult}$-projection of the unique least fixed point that satisfies the guard is equal to $T_\plustimes$.

Thus, we can mimic set quantification over $\nats$ and implement addition and multiplication, which allow us to reduce truth in second-order arithmetic to finite-state satisfiability for $\univar$-free sentences using the function~$\hyperize$ presented in the proof of Theorem~\ref{thm_satcomplexity_lfp_ss}:
\[
\phi' = \phi_\prefs \wedge \exists \pi_\add \in \unidisvar.\ \exists \pi_\mult \in \unidisvar.\ \psi_s \wedge \exists(X_\arith, \smallest, \phi_\arith').\ \hyperize(\phi)
\]
is finite-state satisfiable if and only if $\natsstruct\models\phi$.
\end{proof}

This result, i.e., that finite-state satisfiability for $\univar$-free $\lfpsohyltlfp$  sentences is at least as hard as truth in second-order arithmetic, should be contrasted with the general satisfiability problem for $\univar$-free $\lfpsohyltlfp$ sentences, which is \myquot{only} $\Sigma_1^1$-complete~\cite{sohypercomplexity}, i.e., much simpler. 
The reason is that every satisfiable $\lfpsohyltlfp$ sentence has a countable model (i.e., a countable set of traces). 
This is even true for the formula~$\phi_\prefs$ we have constructed. However, every \emph{finite-state} transition system that satisfies the formula must have uncountably many traces. 
This fact allows us to mimic quantification over arbitrary subsets of $\nats$, which is not possible in a countable model. 
Thus, the general satisfiability problem is simpler than the finite-state satisfiability problem.

Finally, we prove the upper bound for $\lfpsohyltlfp$ model-checking.

\begin{lemma}
\label{lemma_mc_soaeasy}
$\lfpsohyltlfp$ model-checking is polynomial-time reducible to truth in second-order arithmetic. 
\end{lemma}

\begin{proof}
As done before, we will mimic trace quantification via quantification of sets of natural numbers and capture the temporal operators using addition. 
To capture the \myquot{quantification} of fixed points, we use again witness trees. 
Recall that these depend on a variable assignment of the trace variables. 
Thus, in our construction, we need to explicitly handle such an assignment as well (encoded by a set of natural numbers) in order to be able to correctly apply witness trees.

Let $\tsys = (V,E,I,\lambda)$ be a finite transition system. 
We assume without loss of generality that $V = \set{0,1,\ldots, n}$ for some $n \ge 0$.
Recall that $\pair \colon \nats\times\nats \rightarrow\nats$ denotes Cantor's pairing function defined as $\pair(i,j) = \frac{1}{2}(i+j)(i+j+1) +j$, which is a bijection and implementable in arithmetic.
We encode a path~$\rho(0)\rho(1)\rho(2)\cdots$ through $\tsys$ by the set~$\set{\pair(j,\rho(j)) \mid j \in\nats} \subseteq \nats$.
Not every set encodes a path, but the first-order formula
\begin{align*}
\phi_\ispath(Y) = {}&{} \forall x.\ \forall y.\ y > n \rightarrow \pair(x,y) \notin Y \wedge \\
{}&{} \forall x.\ \forall y_0.\ \forall y_1.\ (\pair(x,y_0) \in Y \wedge
 \pair(x,y_1) \in Y) \rightarrow y_0=y_1  \wedge  \\
{}&{} \bigvee_{v \in I} \pair(0,v) \in Y \wedge \\
{}&{} \forall j.\ \bigvee_{(v,v') \in E} \pair(j,v) \in Y \wedge \pair(j+1,v') \in Y
\end{align*}
checks if a set does encode a path of $\tsys$.

Furthermore, fix some bijection~$e \colon \ap \rightarrow\set{0,1,\ldots,\size{\ap}-1}$.
As before, we encode a trace~$t \in (\pow{\ap})^\omega$ by the set~$S_t =\set{\pair(j,e(\prop)) \mid j \in \nats \text{ and } \prop \in t(j)} \subseteq \nats$.
While not every subset of $\nats$ encodes some trace~$t$, the first-order formula 
\[\phi_\istrace(Y) = \forall x.\ \forall y.\ y \ge \size{\ap} \rightarrow \pair(x,y) \notin Y \] checks if a set does encode a trace.

Finally, the first-order formula
\[
\phi_{\tsys}(Y) = \exists Y_p.\ \ispath(Y_p) \wedge \forall j.\ \bigwedge_{\prop\in\ap} \left( \pair(j,e(\prop))\in Y \leftrightarrow \bigvee_{v\colon \prop \in \lambda(v)} \pair(j,v) \in Y_p\right)
\]
checks whether the set~$Y$ encodes the trace of some path through $\tsys$.

As in the proof of Theorem~\ref{thm_satcomplexity_lfp_ss}, we need to encode variable assignments (whose domain is restricted to trace variables) via sets of natural numbers. 
Using this encoding, one can \myquot{update} encoded assignments, i.e., there exists a formula~$\phi_{\update}^\pi(A,A',Y)$ that is satisfied if and only if
\begin{itemize}
    \item the set~$A$ encodes a variable assignment~$\Pi$,
    \item the set~$A'$ encodes a variable assignment~$\Pi'$,
    \item the set~$Y$ encodes a trace~$t$, and
    \item $\Pi[\pi\mapsto t] = \Pi'$.
\end{itemize}

Now, we inductively translate an $\lfpsohyltlfp$ sentence~$\phi$ into a formula~$\arithmetize_\tsys(\phi)$ of second-order arithmetic.
This formula has two free variables, one first-order one and one second-order one.
The former encodes the position at which the formula is evaluated while the latter one encodes a variable assignment (which, as explained above, is necessary to give context for the witness trees).  
We construct $\arithmetize_\tsys(\phi)$ such that 
$\tsys \models \phi$ if and only if $\natsstruct\models (\arithmetize_\tsys(\phi))(0,\emptyset)$, where the empty set encodes the empty variable assignment.

\begin{itemize}
    
    \item $\arithmetize_\tsys(\exists (Y,\smallest,\psi_j^\con).\ \psi) = \arithmetize_\tsys(\psi)$. Here, the free variables of $\arithmetize_\tsys(\exists (Y,\smallest,\psi_j^\con).\ \psi)$ are the free variables of $\arithmetize_\tsys(\psi)$. Thus, we ignore quantification over least fixed points, as we instead use witness trees to check membership in these fixed points.
    
    \item $\arithmetize_\tsys(\forall (Y,\smallest,\psi_j^\con).\ \psi) = \arithmetize_\tsys(\psi)$. Here, the free variables of $\arithmetize_\tsys(\forall (Y,\smallest,\psi_j^\con).\ \psi)$ are the free variables of $\arithmetize_\tsys(\psi)$. 
    
    \item $\arithmetize_\tsys(\exists\pi\in \univar.\ \psi) = \exists Z_\pi.\  \exists A'.\ \phi_{\istrace}(Z_\pi) \wedge \phi_\update^\pi(A,A',Z_\pi) \wedge \arithmetize_\tsys(\psi)$. 
    Here, the free second-order variable of $\arithmetize_\tsys(\exists\pi\in \univar.\ \psi)$ is $A$ while $A'$ is the free second-order variable of $\arithmetize_\tsys(\psi)$. The free first-order variable of $\arithmetize_\tsys(\exists\pi\in \univar.\ \psi)$ is the free first-order variable of $\arithmetize_\tsys(\psi)$.
    
    \item $\arithmetize_\tsys(\forall\pi\in \univar.\ \psi) = \forall Z_\pi.\ \forall A'.\ (\phi_{\istrace}(Z_\pi) \wedge \phi_\update^\pi(A,A',Z_\pi) )\rightarrow \arithmetize_\tsys(\psi)$. 
    Here, the free second-order variable of $\arithmetize_\tsys(\forall\pi\in \univar.\ \psi)$ is $A$ while $A'$ is the free second-order variable of $\arithmetize_\tsys(\psi)$. The free first-order variable of $\arithmetize_\tsys(\forall\pi\in \univar.\ \psi)$ is the free first-order variable of $\arithmetize_\tsys(\psi)$.

    \item $\arithmetize_\tsys(\exists\pi\in \unidisvar.\ \psi) = \exists Z_\pi.\  \exists A'.\ \phi_\tsys(Z_\pi) \wedge \phi_\update^\pi(A,A',Z_\pi) \wedge \arithmetize_\tsys(\psi)$. 
    Here, the free second-order variable of $\arithmetize_\tsys(\exists\pi\in \unidisvar.\ \psi)$ is $A$ while $A'$ is the free second-order variable of $\arithmetize_\tsys(\psi)$. The free first-order variable of $\arithmetize_\tsys(\exists\pi\in \unidisvar.\ \psi)$ is the free first-order variable of $\arithmetize_\tsys(\psi)$.
    
    \item $\arithmetize_\tsys(\forall\pi\in \unidisvar.\ \psi) = \forall Z_\pi.\ \forall A'.\ (\phi_\tsys(Z_\pi) \wedge \phi_\update^\pi(A,A',Z_\pi) )\rightarrow \arithmetize_\tsys(\psi)$. 
    Here, the free second-order variable of $\arithmetize_\tsys(\forall\pi\in \unidisvar.\ \psi)$ is $A$ while $A'$ is the free second-order variable of $\arithmetize_\tsys(\psi)$. The free first-order variable of $\arithmetize_\tsys(\forall\pi\in \unidisvar.\ \psi)$ is the free first-order variable of $\arithmetize_\tsys(\psi)$.

    \item $\arithmetize_\tsys(\exists\pi\in Y_j.\ \psi) = \exists Z_\pi.\ \exists A'.\ \phi_\hastree(A,Z_\pi,j) \wedge \phi_\update^\pi(A,A',Z_\pi) \wedge \arithmetize_\tsys(\psi)$, where $\phi_\hastree(A,Z_\pi,j)$ is a formula of second-order arithmetic that captures the existence of a witness tree for the trace being encoded by $Z_\pi$ being in the fixed point assigned to $Y_j$ w.r.t.\ the variable assignment encoded by $A$. Its construction is similar to the corresponding formula in the proof of Theorem~\ref{thm_satcomplexity_lfp_ss}, but we replace the free third-order variable~$X_D$ encoding the model there by hardcoding the set of traces of $\tsys$ using the formula~$\phi_\tsys$ from above.
   
    Here, the free second-order variable of $\arithmetize_\tsys(\exists\pi\in Y_j.\ \psi)$ is $A$ while $A'$ is the free second-order variable of $\arithmetize_\tsys(\psi)$. The free first-order variable of $\arithmetize_\tsys(\exists\pi\in Y_j.\ \psi)$ is the free first-order variable of $\arithmetize_\tsys(\psi)$.

    \item $\arithmetize_\tsys(\forall\pi\in Y_j.\ \psi) = \forall Z_\pi.\ \forall A'.\ (\phi_\hastree(A,Z_\pi,j) \wedge \phi_\update^\pi(A,A',Z_\pi) )\rightarrow \arithmetize_\tsys(\psi)$.
    Here, the free second-order variable of $\arithmetize_\tsys(\forall\pi\in Y_j.\ \psi)$ is $A$ while $A'$ is the free second-order variable of $\arithmetize_\tsys(\psi)$. The free first-order variable of $\arithmetize_\tsys(\forall\pi\in Y_j.\ \psi)$ is the free first-order variable of $\arithmetize_\tsys(\psi)$.

    \item $\arithmetize_\tsys(\psi_1 \vee \psi_2) = \arithmetize_\tsys(\psi_1) \vee \arithmetize_\tsys(\psi_2)$. Here, we require that the free variables of $\arithmetize_\tsys(\psi_1)$ and $\arithmetize_\tsys(\psi_2)$ are the same (which can always be achieved by variable renaming), which are then also the free variables of $\arithmetize_\tsys(\psi_1 \vee \psi_2)$. 
    
    \item $\arithmetize_\tsys(\neg\psi) = \neg\arithmetize_\tsys(\psi)$. Here, the free variables of $\arithmetize_\tsys(\neg\psi) $ are the free variables of $ \arithmetize_\tsys(\psi)$.
    
     \item $\arithmetize_\tsys(\X\psi) =  \exists i' (i' = i+1) \wedge \arithmetize_\tsys(\psi)$, where $i'$ is the free first-order variable of $\arithmetize_\tsys(\psi)$ and $i$ is the free first-order variable of $\arithmetize_\tsys(\X\psi)$.
     The free second-order variable of $\arithmetize_\tsys(\X\psi)$ is equal to the free second-order variable of $\arithmetize_\tsys(\psi)$.
    
    \item $\arithmetize_\tsys(\psi_1\U\psi_2) =  \exists i_2.\ i_2 \ge i \wedge \arithmetize_\tsys(\psi_2) \wedge \forall i_1.\ (i \le i_1 \wedge i_1 < i_2) \rightarrow \arithmetize_\tsys(\psi_1)$, where $i_j$ is the free first-order variable of $\arithmetize_\tsys(\psi_j)$, 
    and $i$ is the free first-order variable of $\arithmetize_\tsys(\psi_1\U\psi_2)$.
    Furthermore, we require that the free second-order variables of the $\arithmetize_\tsys(\psi_j)$ are the same, which is then also the free second-order variable of $\arithmetize_\tsys(\psi_1\U\psi_2)$. Again, this can be achieved by renaming variables.

    \item $\arithmetize_\tsys(\prop_\pi) = \pair(i,e(\prop)) \in Z_\pi$, i.e., $i$ is the free first-order variable of $\arithmetize_\tsys(\prop_\pi)$. Note that this formula does not have a free second-order variable. For completeness, we can select an arbitrary one to serve that purpose.
    
\end{itemize}
Now, an induction shows that $\tsys\models \phi$ if and only if $\natsstruct$ satisfies $(\arithmetize(\phi))(0,\emptyset)$, where $\emptyset$ again encodes the empty variable assignment.
\end{proof}

\bibliographystyle{plain}
\bibliography{bib}

\end{document}